\documentclass[12pt,a4paper]{article}
\usepackage[margin=1in,nohead]{geometry}
\usepackage[T1]{fontenc}
\usepackage[utf8]{inputenc}
\usepackage{ae}
\usepackage{graphicx}
\usepackage[american]{babel}
\usepackage{amsfonts}
\usepackage{amsthm}
\usepackage{amsmath}
\usepackage{amssymb}
\usepackage{enumerate}
\usepackage{wasysym}
\usepackage{paralist}
\usepackage{multirow}

\newtheorem{Def}{Definition}
\newtheorem{Theorem}{Theorem}
\newtheorem{Lemma}[Theorem]{Lemma}
\newtheorem{Corollary}[Theorem]{Corollary}
\newtheorem{Proposition}[Theorem]{Proposition}
\newtheorem{Conjecture}[Theorem]{Conjecture}

\newcommand{\sgn}{\operatorname{sgn}}
\newcommand{\out}{\operatorname{out}}
\newcommand{\len}{\operatorname{len}}
\newcommand{\capp}{\operatorname{cap}}
\newcommand{\D}{\mathcal{D}}
\newcommand{\dist}{\operatorname{dist}}
\newcommand{\eps}{\varepsilon}
\newcommand{\vE}{\vec{E}}
\newcommand{\EG}{\vec{E}(G)}
\newcommand{\E}{\mathbb{E}}
\newcommand{\F}{\mathcal{F}}
\newcommand{\G}{\mathcal{G}}
\newcommand{\N}{\mathbb{N}}
\renewcommand{\P}{\mathbb{P}}
\newcommand{\R}{\mathbb{R}}
\newcommand{\U}{\mathcal{U}}
\renewcommand{\:}{\colon}
\DeclareMathOperator{\LA}{LA}
\DeclareMathOperator{\RLA}{RLA}
\DeclareMathOperator{\MLA}{MLA}
\DeclareMathOperator{\URG}{URG}
\DeclareMathOperator{\URN}{URFN}
\DeclareMathOperator{\FIID}{FIID}

\title{Local algorithms for the maximum flow and minimum cut in bounded-degree networks}
\author{Endre Cs\'oka%
\thanks{Alfr\'ed R\'enyi Institute of Mathematics. Supported by the NRDI grant KKP~138270.%
}, Andr\'as Pongr\'acz%
\thanks{Alfr\'ed R\'enyi Institute of Mathematics. Supported by the NRDI grant KKP~138270. \newline
Competing interests: The authors declare none.%
}}
\date{}

\begin{document}

\maketitle

\begin{abstract}

We show a deterministic constant-time local algorithm for constructing an approximately maximum flow and minimum fractional cut in multisource-multitarget networks with bounded degrees and bounded edge capacities. 
Locality means that the decision we make about each edge only depends on its constant radius neighborhood. 
We show two applications of the algorithms: one is related to the Aldous-Lyons Conjecture, and the other is about approximating the neighborhood distribution of graphs by bounded-size graphs. 
The scope of our results can be extended to unimodular random graphs and networks. 
As a corollary, we generalize the Maximum Flow Minimum Cut Theorem to unimodular random flow networks. 

\end{abstract}

\section{Introduction}

A deterministic local algorithm ($\LA$) on a graph means that we make a decision at each vertex or edge depending on the isomorphism class of its constant-radius neighborhood of that vertex or edge. 
It is a special case of the more commonly used notion of a random local algorithm ($\RLA$), where we assign independent uniform random seeds from $[0, 1]$ to each vertex, and we make a decision at each vertex or edge depending on its constant-radius seeded neighborhood. 
For example, choosing the set of vertices with higher seeds than all of its neighbors is an $\RLA$ that finds an independent set of expected size $\sum\limits_{x \in V} \frac{1}{\deg(x) + 1}$. 

Local algorithms were defined by Linial \cite{Linial} as distributed algorithms whose complexity is measured by the number of synchronized rounds regardless of the computational time and space. 
The notion of a local algorithm as defined above means local algorithm in Linial's sense with a constant number of rounds. 

In this paper, we consider graphs with degrees bounded by a constant $d$. 
We note that in this case, $\LA$ and $\RLA$ coincide with deterministic and randomized constant-time distributed algorithms, respectively. 
Research on local algorithms was pioneered by Angluin \cite{Angluin}, Linial \cite{Linial}, and Naor and Stockmeyer \cite{NaSt}. 
Angluin \cite{Angluin} studied the limitations of anonymous networks without any unique identifiers. 
Linial \cite{Linial} proved some negative results for the case where each node has a unique identifier. 
Naor and Stockmeyer \cite{NaSt} presented the first nontrivial positive results. 
For more about local algorithms, see the survey paper by Suomela \cite{Suomela}.

A mixed local algorithm ($\MLA$) is a probability distribution on local algorithms. 
Equivalently, we draw a global random seed $g$ from a uniform distribution on the unit interval, which can be accessed by all vertices. 
The decision depends on the constant radius neighborhood of the vertex or edge, and the global seed $g$. 

For typical problems, we may only expect approximate solutions from local algorithms rather than strictly optimal ones. 
For example, we say that we can find an almost maximum independent set if for each $\eps > 0$, there exists a local algorithm that outputs an independent set, and the size of this set is at most $\eps |V(G)|$ less than the size of the maximum independent set, where $V(G)$ is the vertex set of the graph. 

It is well-understood that randomness can be a useful addition to local algorithms, notably to break the symmetry. 
For instance, a deterministic local algorithm must make the same decision on every vertex or edge of a vertex- or edge transitive graph, say a cycle. 
Hence, a correct $\LA$ (or even a correct $\MLA$) aimed at finding an independent set always chooses the empty set on a nonempty transitive graph: it must put either no vertex or all vertices in the set, and the latter is an incorrect option for an independent set. 
However, we already mentioned a much better $\RLA$ based on a trivial idea that finds an independent set with expected relative size $1/3$ on a cycle, or in general at least $1/(d+1)$ if $d$ is the maximum degree of the graph. 

It is perhaps less obvious how mixing could be advantageous. 
In the present paper, we provide an $\MLA$ that finds a nearly minimum cut in graphs (in fact, in networks). 
In the parallel paper \cite{Csoka}, it is shown that this is not possible to do by using an $\RLA$: there is a positive barrier $C(d)$ for the relative error depending only on the maximum degree $d$ that cannot be broken by any $\RLA$. 
That is, no $\RLA$ finds a cut whose relative size is closer to the optimum than $C(d)$ for all graphs of maximum degree $d$. 
For a more detailed comparison of the power of different variants of local algorithms, see \cite{Csoka}.

The main results of the present paper is the construction of a deterministic local algorithm for finding a nearly maximum flow and for finding a nearly minimum fractional cut in networks with bounded cost function on bounded-degree graphs. 
That is, for any $\eps>0$, $M>0$ and $d\in \mathbb{N}$ there is an $\LA$ that computes a flow in any network with cost bounded by $M$ and degrees bounded by $d$, whose relative value is at most $\eps$ apart from the maximum, and similarly an $\LA$ that computes a fractional cut whose relative value is at most $\eps$ apart from the minimum. 
The above-mentioned $\MLA$ for finding a nearly minimum cut is a mixture of $\LA$s derived from the above $\LA$ for nearly minimum fractional cut and a global seed $g$ with a uniform distribution on $[0,1]$. 
For those who are interested in standard (non-distributed) algorithms rather than the distributed local algorithms considered in the present paper, these results provide a linear-time algorithm for finding a nearly maximum flow and a nearly minimum cut in the bounded degree and bounded cost function setting. 
Namely, the local algorithm can be executed in all vertices of the graph (sequentially rather than in parallel), each execution requiring $O(1)$ time, which yields an $O(|V(G)|)$ total runtime. 

As the error rate is normalized by the number of nodes, these results are only meaningful if we allow any number of sources and targets. 
If we only allowed 1, or even $o(n)$ sources or targets, then the trivial algorithm that constructs the all-zero flow and the empty cut for large enough graphs (dealing with small networks by using a catalogue) would solve the problem up to any prescribed relative error $\eps>0$. 
Hence, even though we need not add any technical assumptions to our theorems because of this observation, it is still better to imagine large networks whose sources and sinks make up at least a constant proportion of the nodes when reading the assertions and proofs in this paper. 


Local algorithms are useful for parameter testing, as well. 
We give a brief overview about parameter testing and the theory of very large graphs; for more details, cf. the survey paper and book \cite{Lovasz, Lovaszbook} by Lov\'asz. 
Parameter testing is an important concept in the theory of bounded degree graphs \cite{BeSchSh, BObT, CzuShaSo, Elek2, Elek, LOW, MaRo, NgOn}. 
For a graph parameter, a tester is an algorithm which gets the constant radius neighborhoods of a constant number of random nodes as input, and outputs a number that estimates the parameter. 
We call a parameter testable if for each $\eps > 0$, there exists a tester which estimates the true parameter with at most $\eps$ error with probability at least $1 - \eps$.

Many of the examined parameters come from optimization problems. 
The relative size of the maximum matching, the maximum independent set, or the minimum vertex cover normalized by the number of nodes, are standard examples. 
In these problems, if we have a random local algorithm which provides an almost optimal structure, e.g.\ an almost maximum independent set, then it is easy to make a tester for the parameter. 
Nguyen and Onak \cite{NgOn} proved the testability of several problems using randomized local algorithms.


In the final section, we generalize our results from finite to unimodular random graphs and networks. 
The 
As a corollary, we obtain the highly nontrivial generalization of the Maximum Flow Minimum Cut Theorem to unimodular random flow networks. 
Furthermore, we show an application related to the Aldous-Lyons Conjecture. 
Namely, we prove that if the conjecture is either true or very far from being true. 
We note that since the first draft of this paper was created, Lovász has proved a very general Maximum Flow Minimum Cut Theorem. 
In that paper~\cite{Lovaszmeasure}, flows and cuts were defined on Markov spaces, and it was shown that there is a measurable minimum cut whose value coincides with that of a flow. 
The proof relies on classical results in functional analysis such as the Hahn-Banach theorem, making it non-constructive. 
In contrast, the result of the present paper, albeit in a restricted scope compared to~\cite{Lovaszmeasure}, are constructive, describing local algorithms that compute asymptotically optimal flows and fractional cuts. 
It is unknown whether such local algorithms exist beyond the realm of unimodular random flow networks with bounded degree. 

\section{Model and results}

Throughout the paper, $d$ is a fixed positive integer, and degrees in graphs are bounded by $d$. 
An input network $N = (G, c)$ consists of a graph $G$ and a capacity function $c$ defined on the edges. 
The graph $G = (S, R, T, \vE)$ has maximum degree at most $d$, and its vertices are separated into the disjoint union of the sets $S$ (source), $R$ (regular) and $T$ (target). The binary relation $\vE$ is the set of directed edges of $G$. 
It is symmetrical, that is, $(a, b) \in \vE \Leftrightarrow (b, a) \in \vE$. 
We have a nonnegative capacity function $c \: \vE \rightarrow [0, \infty)$. 
The total capacity of all edges is denoted by $\capp(N) = \sum\limits_{e \in \vE} c(e)$. 
We usually work under the condition that flow capacities are bounded by some fixed constant $M>0$. 
That is, $c(e)\leq M$ for every edge $e\in \vE$. 

Throughout the paper we use the terms "graph", "path" and "edge" in the directed sense. 
Let $V = V(G) = V(N) = S \dot\cup R \dot\cup T$, $|V| = n$, $\out(A) = \big\{(a, b) \in \vE \ \big|\ a \in A, b \notin A \big\}$, and $\out(v) = \out\big(\{v\}\big)$. 
For an edge $e = (a, b)$, let $-e = (b, a)$ denote the reversed edge.

A function $f\: \vE \rightarrow \mathbb{R}$ is called a flow if it satisfies the following conditions.
\begin{align}
\forall e \in \EG\:&\,\,\,\, f(-e) = -f(e), \label{antisymmetry}
\\ \forall e \in \EG\:&\,\,\,\, f(e) \le c(e), \label{capacity}
\\ \forall r \in R\:&\,\,\,\, \sum_{e \in \out(r)} f(e) = 0. \label{regeq}
\end{align}
The value of a flow $f$ is
\begin{equation} \label{flownorm}
\big\|f\big\| = \sum_{e \in \out(S)} f(e).
\end{equation}
Denote a maximum flow by $f^*(N)$, or simply by $f^*$ is there is no ambiguity about which network is considered.

A partition $V=X\cup(V\setminus X)$ with $T \subseteq X \subseteq T \cup R$ is called a cut. 
We can refer to the cut by the subset $X$ in the partition containing the targets. 
As fractional cuts play an important role in the present paper, it is often more convenient for us to represent $X$ as an indicator function. 
Hence, the alternative definition of a cut is a function $X: V\rightarrow \{0,1\}$ such that $X\upharpoonright_S \equiv 0$ and $X\upharpoonright_R \equiv 1$.
The value of a cut is
\begin{equation} \label{cutnorm}
\big\|X\big\| = \big\|X\big\|_N = \sum\limits_{e \in \out(V\setminus X)} c(e).
\end{equation}
The Maximum Flow Minimum Cut Theorem~\cite{FoFu} states that
\begin{equation} \label{maxflowmincut}
\min\limits_{T \subseteq X \subseteq T \cup R} \big\|X\big\| = \big\|f^*\big\|.
\end{equation}

We define a \emph{fractional cut} of a network $N$ as a function $\tilde{X}\: V(N) \rightarrow [0, 1]$ such that
\begin{equation} \label{fraccut}
\forall s \in S\: \tilde{X}(s) = 0 \text{ and } \forall t \in T\: \tilde{X}(t) = 1.
\end{equation}
The value of a fractional cut is defined as
\begin{equation} \label{fracnorm}
\big\|\tilde{X}\big\| = \big\|\tilde{X}\big\|_N = \sum_{(a, b) \in \vE(N)} c\big((a, b)\big) \max\big(0,\ \tilde{X}(b) - \tilde{X}(a)\big).
\end{equation}

Clearly, every cut $X: V(N)\rightarrow \{0,1\}$ is a fractional cut. 
There is no clash between the value of $X$ as a fractional cut and the original definition of the value of $X$ as a cut: 
\begin{equation*}
\big\|X\big\| \mathop{=}^{\eqref{fracnorm}} \sum_{(a, b) \in \vE(G)} c\big((a, b)\big) \max\big(0,\ X(b) - X(a)\big) \mathop{=} \sum\limits_{(a, b) \in \out(V\setminus X)} c\big((a, b)\big) \mathop{=}^{\eqref{cutnorm}} \big\|X\big\|.
\end{equation*}
Although fractional cuts properly generalize cuts, the value of any fractional cut is a convex combination of values of cuts. 
If $\tilde{X}$ is a fractional cut, then for all $u \in (0, 1)$, $\tilde{X}[u] = \big\{v \in V \big| \tilde{X}(v) \geq u \big\}$ is a cut, and for a uniform random $u$ from $(0, 1)$,
\begin{equation*}
\E_u\Big( \big\|\tilde{X}[u]\big\| \Big) \mathop{=}^{\eqref{cutnorm}} \E_u\Big( \sum\limits_{e \in \out(V\setminus\tilde{X}[u])} c(e) \Big) = \sum\limits_{e \in \vE} \mathbb{P} \Big(e \in \out \big(V\setminus\tilde{X}[u] \big) \Big) c(e)
\end{equation*}
\begin{equation} \label{fractocut}
= \sum\limits_{(a, b) \in \vE} \mathbb{P}(\tilde{X}(a) < u \le \tilde{X}(b)) c((a,b)) = \sum\limits_{(a, b) \in \vE} c\big((a, b)\big) \max\big(0,\ \tilde{X}(b) - \tilde{X}(a)\big) \mathop{=}^{\eqref{fracnorm}} \big\|\tilde{X}\big\|.
\end{equation}

These two observations imply that it makes no difference to consider the minimum fractional cut instead of the minimum cut, that is, $\min_X \big\|X\big\| = \min_{\tilde{X}} \big\|\tilde{X}\big\|$.

\subsection{Local flow and local cut algorithms}

The rooted $r$-neighborhood of a vertex $v$ or edge $e$, denoted by $B_r(v) = B_r(G, v)$ and $B_r(e)$, means the (vertex- or edge-)rooted induced subnetwork of the vertices at distance at most $r$ from $v$ or $e$, rooted at $v$ or $e$, respectively. 
The set of all possible $r$-neighborhoods (up to isomorphism) are denoted by $B_r$ and $
B_r^{(2)}$, respectively. 
As we restrict our attention to graphs all of whose vertices have degree at most $d$, there are only finitely many possible underlying rooted graph structures in $B_r$ and $B_r^{(2)}$; with infinitely many possible different cost functions, of course, forming a compact set if we assume that a universal bound $c(e)\leq M$ is given. 
A function $F\: B_r^{(2)} \rightarrow \mathbb{R}$ is called a local flow algorithm, if for each network $N$, the induced mapping $F(N) = \big(e \mapsto F(B_r(e))\big)$ is a flow. 
Similarly, a function $C\: B_r \rightarrow \{0,1\}$ is called a local cut algorithm, if for each network $N$ the induced function $C(N) = \big(v \mapsto C(B_r(v))\big)$ is a cut. 
A function $\tilde{C}\: B_r \rightarrow [0, 1]$ is called a local fractional cut algorithm, if for each network $N$ the induced function $\tilde{C}(N) = \Big( v \mapsto \tilde{C}\big(B_r(v)\big)\Big)$ is a fractional cut. 

\begin{Theorem} \label{localg}
For all $\eps > 0$, there is an $r\in \N$ and a local flow algorithm $F\: B_r^{(2)} \rightarrow \mathbb{R}$ such that 
\begin{equation} \label{Thm1eq}
\big\|F(N)\big\| \ge \big\|f^*(N)\big\| - \eps \capp(N).
\end{equation}
In particular, assuming a universal bound $M$ on the expected capacity, for all $\eps > 0$, there is an $r\in \N$ and a local flow algorithm $F\: B_r^{(2)} \rightarrow \mathbb{R}$ such that 
\begin{equation} \label{Thm1eqbounded}
\big\|F(N)\big\| \ge \big\|f^*(N)\big\| - \eps n.
\end{equation}
\end{Theorem}

\begin{Corollary}\label{cor:integerflow}
Let $x$ be a mapping that assigns an independent uniform random seed in $[0,1]$ to the vertices of a network. 
Assuming a universal bound $M$ on the  expected capacity and that all capacities are integers, for all $\eps > 0$, there is an $r\in \N$ and a random local flow algorithm $F\: B_r^{(2)} \rightarrow \mathbb{R}$ defined on the seeded neighborhoods such that 
\begin{equation} \label{Thm1eqbounded}
\E_x\big(\big\|F(N)\big\|\big) \ge \big\|f^*(N)\big\| - \eps n.
\end{equation}
and such that all values of the output flow are integers. 
\end{Corollary}

\begin{Proposition} \label{localgfraccut}
Assuming a universal bound $M$ on the capacity function, for each $\eps > 0$ there exists a local fractional cut algorithm such that 
\begin{equation}
\big\|\tilde{C}(N)\big\| \le \|f^*(N)\| + \eps n.
\end{equation} 
\end{Proposition}

In other words, maximum flow is approximable in $\LA$, maximum integer flow is approximable in $\RLA$, and minimum fractional cut is approximable in $\LA$. 
As an easy consequence, we obtain an algorithm in $\MLA$ that approximates a minimum cut. 

\begin{Theorem} \label{localgcut}
Assuming a universal bound $M$ on the capacity function, for all $\eps > 0$, there is an $r\in \N$ and a probability distribution $\D$ of local cut algorithms $C\: B_r \rightarrow \{0,1\}$ such that 
\begin{equation*}
\E_{C \in \D} \big\|C(N)\big\| \le \big\|f^*(N)\big\| + \eps n.
\end{equation*}
\end{Theorem}

We mention that with a more accurate calculation, the ideas of this paper could be refined to yield an algorithm for each problem using radius $\approx d^{1 / \eps}$ (if $\eps$ is small enough).

\begin{Corollary} \label{testable}
Assuming a universal bound $M$ on the capacity function, the network parameter $\big\|f^*(N)\big\| / n$ is testable. 
Namely, for all $\eps > 0$ there exist $k, r \in \mathbb{N}$ and a function $g\: B_r^k \rightarrow \mathbb{R}$ satisfying that if the vertices $v_1, v_2, \ldots, v_k$ are chosen independently with uniform distribution, then
\begin{equation*}
\E\bigg( \bigg|\frac{\|f^*(N)\|}{n} - g\big(B_r(v_1), B_r(v_2), \ldots, B_r(v_k)\big)\bigg| \bigg) < \eps.
\end{equation*}
\end{Corollary}

We note that for all $\eps > 0$, having an approximation with less than $\eps n$ error in expectation is a stronger requirement than having this error with at least $1 - \eps$ probability. 
But these conditions are equivalent if the error is bounded.

\subsection{Applications on distributions of neighborhoods}\label{subsec:AL}

Let $\G$ denote the set of all graphs with degrees bounded by $d$.
Let $s_r(G) \in \R^{B_r}$ denote the distribution of the $r$-neighborhood of a random vertex of a graph $G \in \G$.
We call a family $\F \subseteq \G$ of graphs \textbf{nice} if it is union-closed and closed under taking spanned subgraphs, excluding the empty graph. 
Formally, $G_1, G_2 \in \F \Rightarrow G_1 \cup G_2 \in \F$, $G_1 \subseteq^* G \in \F \Rightarrow G_1 \in \F$ and $\emptyset \notin \F$, where $\subseteq^*$ denotes nonempty spanned subgraph. 
Let us denote the closure of the set of all $r$-neighborhood distributions in $\F$ by 
\begin{equation} \label{Ddef}
S_r(\F) = cl\big\{s_r(G) \ \big|\ G \in \F \big\}.
\end{equation}

\begin{Theorem} \label{separation}
Assume that $s_r(G_0) \notin S_r(\F)$ holds for a nice family $\F$ of graphs, a graph $G_0 \in \G$ and $r\in \N$.
Then for all $\eps > 0$, there exists an $r'\in \N$, a subset $B' \subset B_{r'}$, and a graph $G_1 \in \G$ such that 
\begin{align*}
\P\big(s_{r'}(G_1) \in B'\big) &> 1 -\eps, \,\, and\\
\forall G \in \F\: \,\, \P\big(s_{r'}(G) \in B'\big) &< \eps.
\end{align*}
\end{Theorem}

This theorem is shown in a slightly more general form in Subsection~\ref{subs:pumping}, giving an explicit bound on $r' = r'(r, \eps, \delta)$, where $\delta$ is a certain kind of distance of $s_r(G_0)$ from $S_r(\F)$.

Lovász \cite{Lovaszbook} asked to find, for every radius $r \in \N$ and error bound $\eps > 0$, an explicit $n \in \N$ such that the $r$-neighborhood distribution of each graph can be $\eps$-approximated by a graph of size at most $n$. Formally,
\begin{align} \label{ndef}
\forall G \in \G \: \exists G' \in \G \: \big|V(G')\big| \le n,
&& \big\|s_r(G) - s_r(G')\big\|_1 < \eps.
&&&&&&&&&&&&&&&&&&&&&&&&&
\end{align}

Without requiring an effective bound, it is obvious that such a bound $n$ and family of graphs exist, as the set $S_r(\F)$ is compact. 
It is still open whether, say, a recursive function exists that computes such a bound. 
It is also unclear how to compute the graph $G'$ from $G$; see \cite{Lovaszbook} for details.

Instead of $\big|V(G')\big| \le n$, we only require that \emph{each component} of $G'$ has size at most $n$. This version is very close to the original question, because the $r$-neighborhood distribution of such a graph is a convex combination of the $r$-neighborhood distributions of the components, and each convex combination can be approximated by a graph with a bounded number of small components.

Let $ n = n(r, \eps)$ denote the smallest value of $n$ satisfying the modified conditions. 
The following corollary shows that if there exists an arbitrarily large error bound $\lambda < 1$ such that for all $r$, we can find an explicit upper bound on $n(r, \lambda)$, then it provides explicit upper bounds on $n(r, \eps)$ for all $r \in \N$ and $\eps > 0$, as well.

\begin{Corollary} \label{smallgraph}
For all $r \in \N$, $0 < \delta$, and $\lambda < 1$, there exists an $r'\in\N$ such that $n(r, \delta) \le n(r', \lambda)$. 
\end{Corollary}

\section{Proofs}

\subsection{Proof of Theorem~\ref{localg} and Corollaries~\ref{cor:integerflow} and \ref{testable}}

An augmenting path of a flow $f$ is a directed path $u = (e_1, e_2, \ldots, e_k)$ from $S$ to $T$ with $f(e_i) < c(e_i)$ for each edge $e_i$. The capacity of $u$ means $\capp(u) = \capp(u, f) = \min\limits_{i}(c(e_i) - f(e_i))$. We identify an augmenting path $u$ with the flow $u\: \EG \rightarrow \mathbb{R}$, $u(e) =
\begin{cases}
1\text{ if }\exists i\: e = e_i\\
-1\text{ if }\exists i\: e = - e_i\\
0\text{ otherwise}
\end{cases}$
which we also call path-flow. Augmenting on such a path $u$ means the incrementation of $f$ by $\capp(u) \cdot u$. 
The length of a path $u$ is the number of its edges, denoted by $\len(u)$. 

\begin{Lemma} \label{approx1}
If a flow $f$ has no augmenting path of length at most $\ell$, then
\begin{equation} \label{approx1eq}
\big\|f\big\| \ge \big\|f^*\big\| - \frac{\capp(N)}{\ell}.
\end{equation}
\end{Lemma}

\begin{proof}
The function $f^* - f$ is a flow on the network $(G, \bar{c})$, where $\bar{c}(e) = c(e) + c(-e)$. 
This flow can be decomposed into the sum of path-flows $u_1, u_2, \ldots, u_q$ and a circulation $u_0$ that follow the directions of the flow $f^* - f$, i.e., for every $i \in \{0, \ldots, q\}$ and $e \in \EG$ we have $\sgn(u_i(e)) \in \{0, \sgn((f^*-f)(e))\}$. 
For example, we can do it by the Ford-Fulkerson algorithm \cite{FoFu} on the network $\big((S, R, T, \{e \in \EG \: (f^*-f)(e) > 0\}), \bar{c}\big)$. 
Thus
\begin{equation*}
\big\|f^*\big\| - \big\|f\big\| = \big\|f^* - f\big\| = \sum_{i=1}^q \big\|u_i\big\| = \frac{1}{2\ell} \sum_{i=1}^q 2\ell \big\|u_i\big\| \le \frac{1}{2\ell} \sum_{i=1}^q \sum_{e \in \EG} \big\|u_i(e)\big\|
\end{equation*}
\begin{equation*}
= \frac{1}{2\ell} \sum_{e \in \EG} \Big\|\sum_{i=1}^q u_i(e)\Big\| \le \frac{1}{2\ell} \sum_{e \in \EG} \bar{c}(e) \le \frac{2\capp(N)}{2\ell} = \frac{\capp(N)}{\ell}. \qedhere
\end{equation*}
\end{proof}

The next lemma is well-known. 
It is crucial to us as this lemma guarantees that the local flow algorithm constructed in the proof of Theorem~\ref{localg} halts. 

\begin{Lemma}\label{lem:nonewshort}
If a flow $f$ has no augmenting path $u$ with $\len(u)<\ell$, then augmenting on a path of length $\ell$ does not create a new augmenting path of length at most $\ell$.
\end{Lemma}

\begin{proof}
Let the residual graph of a network $N = (G, c)$ with respect to a flow $f$ be the graph $G_f = \Big(S(G), R(G), T(G), \big\{e \in \EG \big| f(e) < c(e) \big\}\Big)$. 
Then the augmenting paths of $G$ can be identified with the paths in $G_f$ from $S$ to $T$. 
Therefore, the length of the shortest path in $G_f$ from $S$ to $T$ is $k$. 
Let the forwardness of an edge of $N$ be defined as the difference of the distances of its endpoint and starting point from $S$ in $G_f$. 
Augmenting on a shortest path adds only such edges to the residual graph on which $f$ decreases, which are the reverse edges of the path. 
All these edges have forwardness $-1$ (calculated before augmenting). 
So if a path becomes an augmenting path at this augmenting step, then all its edges have forwardnesses at most 1 and contain an edge with forwardness $-1$. 
As the total forwardness is $k$, the length of the path is at least $k + 2$. 
\end{proof}

Let us fix a length $\ell$, and let us call the paths of length at most $\ell$ \textbf{short paths}.
Let us choose a uniform random seed $x(u)\in [0,1]$ independently for every short path $u$. 
We may assume that these seeds are all different, as it is an almost sure event. 
The label of $u$ is $\len(u)+x(u)$. 
Hence, if we order the short paths by label, then shorter ones always precede longer paths. 
In other words, we first partially ordered the short paths by length, and then extended the partial order to a total order by randomly sorting paths of the same length. 

We define a \textbf{chain} as a sequence $u_1, u_2, \ldots,  u_s$ of short paths with increasing random labels such that for all $i \in \{1, 2, \ldots, s-1\}$ there exists a common undirected edge of $u_i$ and $u_{i+1}$ (henceforth: these paths intersect each other). 

\begin{Lemma} \label{spread}
For each $\ell \in \mathbb{N}$ and $\eps > 0$ there exists a $q = q(\ell, \eps) \in \mathbb{N}$ such that for every graph $G$ and undirected edge $e$ of $G$, 
the probability that there exists a chain $u_1, u_2, \ldots, u_q$ for which $u_q$ contains $e$ is at most $\eps$.
\end{Lemma}

\begin{proof}
There exists an upper bound $z = z(\ell)$ for the number of short paths that intersect a given short path. 
Hence, there are at most $z^q$ sequences of short paths $u_1, u_2, \ldots, u_q$ for which $e$ is in $u_q$ and $\forall i \in \{1, 2, \ldots, q-1\}$ the path $u_i$ intersects $u_{i+1}$. 
All these sequences contain $\lceil q/\ell \rceil$ paths of the same length. 
The ordering of their labels is drawn randomly from $\lceil q/\ell \rceil !$ permutations. 
So the probability that the labels are decreasing is $1 / \lceil q/\ell \rceil !$. This event is necessary for the sequence to be a chain. Denote the number of chains in the lemma by the random variable $X$ (with respect to the random labelling). Then 

\begin{equation*}
P(X \ge 1) \le \E(X) \le \frac{z^q}{\lceil q/\ell \rceil !} \rightarrow 0 \text{ as } q \rightarrow \infty ,
\end{equation*}
which proves the lemma for some large enough number $q$.
\end{proof}

We prove Theorem~\ref{localg} after a series of lemmas. 
Consider the variant of the Edmonds--Karp algorithm where we augment on the short paths in the order of their label. 
In other words, we start from the empty flow, we take all short paths $u$ in increasing order of $\len(u) + x(u)$, and with each path, we increase the actual flow $f$ by $\capp(f, u) \cdot u$. 
We denote this algorithm by $A_1$ and the resulting flow by $f_1 = f_1(N, x)$.

Consider now the variant of the previous algorithm where we skip augmenting on each path which can be obtained as the last element of any chain of length $s$. 
We denote this algorithm by $A_2$ and the resulting flow by $f_2 = f_2(N, x)$. The next lemma shows that $f_2$ is a local algorithm.

\begin{Lemma} \label{f2local}
For each edge $e$ and random seed $x$, we have
\begin{equation*}
f_2(N, x)(e) = f_2 \big( B_{s\ell}(e), (x\upharpoonright_{V(B_{s\ell}(e))})\big)(e).
\end{equation*}
\end{Lemma}

\begin{proof}
Let us execute the two algorithms $f_2(N, x)$ and $f_2 \big( B_{s\ell}(e), (x\upharpoonright_{V(B_{s\ell}(e))})\big)$ in parallel. 
We adjust the timeline of the latter one so that whenever $f_2(N, x)$ augments on a short path $u$ that is not fully contained in $B_{s\ell}(e)$, then we make an idle step. 
If at a point, the two flows differ at an edge $\tilde{e} \in \vE(B_{s\ell}(e))$, then there must have been a path $u$ through $\tilde{e}$ on which the two algorithms augmented by different values. 
There are three possible reasons to it:
\begin{enumerate}
\item $u$ is not in $B_{s\ell}(e)$;
\item $u$ can be obtained as the last term of some chain of length $s$ in $G$, but not in $B_{s\ell}(e)$;
\item $u$ has an edge $e'$ in $B_{s\ell}(e)$ at which the values of the two flows were different before taking $u$.
\end{enumerate}

Assume that at the end, the two flows are different on $e$. 
Using the previous observation initially with $\tilde{e} = e$, let us take a short path through $\tilde{e}$ on which the two augmentations were different, and consider which of the three reasons occurred. 
As long as it is the third one, repeat the step with setting $\tilde{e}$ as the $e'$ of the previous step. 
Since by each step we jump to an earlier point of the execution of the algorithms, we must get a reason different from the third one eventually. 
Denote the considered paths, in order of appearance during the execution of $f_2(N, x)$, by $u_1, u_2, \ldots,  u_t$, where $u_1$ is the path where reason 1 or 2 applies, and $u_2, \ldots,  u_t$ were considered for reason 3. 
(That is, in the above description, we found $u_t$ first and $u_1$ last.)

Assume that reason 1 applies to $u_1$.
The set of all edges of the $t$ paths $u_1, u_2, \ldots,  u_t$ is connected. 
This set contains at most $t\ell$ edges. 
On the other hand, it must contain $e$ and an edge at least $s\ell$ away from $e$, so $t\ell > s\ell$, and consequently $t > s$. Thus the last $s$ short paths $u_{t-s+1}, u_{t-s+2}, \ldots, u_t$ is a chain with a connected edge set with size at most $s\ell$ containing $e$, so this chain is in $B_{s\ell}(e)$. 
Therefore, neither executions should have been augmented on $u_t$, a contradiction. 

Now assume that reason 2 applies to $u_1$. 
Then choose a chain of length $s$ as described in reason 2, ending with $u_1$, and append the sequence $u_2, u_3, \ldots,  u_t$ to it. 
The last $s$ short paths of the sequence obtained forms a chain that yields the same contradiction as in the previous case.
\end{proof}

Let $f_2$ be the output of $A_2$ and $\ell = \lceil 2/\eps\rceil$. 
With $q$ being the function guaranteed by Lemma~\ref{spread}, let $s = q(\ell, \eps/4)$. 
We show the following chain of inequalities. 
\begin{equation} \label{eq1}
\E\big(\big\|f_2\big\|\big) \ge \E\big(\big\|f_1\big\|\big) - \frac{\eps}{2} \capp(N) \mathop{\ge}^{\eqref{approx1eq}} \big\|f^*\big\| - \eps \capp(N)
\end{equation}

First we prove the second inequality of \eqref{eq1}.
As $f_1$ contains no short augmenting path, by using Lemma~\ref{approx1} we obtain
\begin{equation*}
\big\|f_1\big\| \mathop{\ge}^{\eqref{approx1eq}} \big\|f^*\big\| - \frac{\capp(N)}{\ell} \geq \big\|f^*\big\| - \frac{\capp(N)}{2/\eps} = \big\|f^*\big\| - \frac{\eps}{2} \capp(N).
\end{equation*}

To prove the first inequality of \eqref{eq1}, we need the following lemma.

\begin{Lemma} \label{diff}
If $f_1(N,x)(e) \neq f_2(N,x)(e)$, then there exists a short path through $e$ which is the last term of a chain of length $s$.
\end{Lemma}

\begin{proof}
Let us consider the executions of $A_1$ and $A_2$ in parallel so that at the same time these take the same augmenting path.
If at a point of the executions, the two flows differ in an edge $\tilde{e}$, then there must have been a path $u$ through $\tilde{e}$ on which the two algorithms augmented by different values. 
There are two possible reasons to it:
\begin{enumerate}
\item $u$ is the last term of a chain of length $s$;
\item $u$ has an edge $e'$ on which the values of the two flows were different before taking $u$.
\end{enumerate}
Assume that at the end, the two flows are different at $e$. 
Using the previous observation, let us take a path $u$ through $\tilde{e}$ on which the two augmentation were different, and consider which of the two reasons occurred. 
As long as it is reason 2, repeat the step with choosing $\tilde{e}$ as the $e'$ of the previous step. 
Since by each step we jump to an earlier point of the execution of the algorithms, we must get reason 1 eventually. 
Denote the paths considered during the process by $u_1, u_2, \ldots, u_t$. 
The path $u_1$ is the last term of a chain of length $s$. 
Hence, by appending the sequence $u_2, \ldots, u_t$ and keeping the last $s$ elements of the chain obtained, we arrive at a chain of length $s$ whose last term is $u_t$.
\end{proof}

\begin{proof}[Proof of Theorem~\ref{localg}]
Clearly $f_2\leq f_1$. 
If $f_1(N,x)(e) \neq f_2(N,x)(e)$, then by Lemma~\ref{diff}, there exists a chain of length $s=q(\ell, \frac{\eps}{4})$. 
By Lemma~\ref{spread} this has probability at most $\frac{\eps}{4}$, and when it occurs, the difference $f_1(N,x)(e) - f_2(N,x)(e)$ is at most $2c(e)$. 
Therefore,
\begin{equation*}
\E_x\big(\big\|f_1(N,x)\big\|\big) - \E_x\big(\big\|f_2(N,x)\big\|\big) = \E_x\big(\big\|f_1(N,x) - f_2(N,x)\big\|\big) =
\end{equation*}
\begin{equation*}
= \E_x\Big(\sum_{e \in \out(S)} \big(f_1(N,x)(e) - f_2(N,x)(e)\big)\Big) = \sum_{e \in \out(S)} \E_x\big(f_1(N,x)(e) - f_2(N,x)(e)\big) \le 
\end{equation*}
\begin{equation*}
\le \sum_{e \in \out(S)} \frac{\eps}{4} \cdot 2c(e) = \frac{\eps}{2} \sum_{e \in \out(S)}c(e) \le \frac{\eps}{2} \capp(N),
\end{equation*}
which finishes the proof of \eqref{eq1}.

The function $\bar{f}_2(e) = \E(f_2(N,x)(e))$ is a convex combination of flows, hence it is a flow, and $\big\|\bar{f}_2\big\| = \E(\big\|f_2(N,x)\big\|) \ge \big\|f^*\big\| - \eps \capp(N)$. 
Furthermore, $\bar{f}_2(e)$ depends only on $B_{s\ell}(e)$, so it can be calculated by a local algorithm. Consequently, this algorithm satisfies the first inequality \eqref{Thm1eq} in Theorem~\ref{localg}. 

The second inequality \eqref{Thm1eqbounded} obviously holds for the local flow algorithm constructed above by substituting $\eps/(dM)$ into $\eps$ in \eqref{Thm1eq}, since $\capp(N)\leq Mdn$.
\end{proof}

\begin{proof}[Proof of Corollary~\ref{cor:integerflow}]
We claim that the flow $f_2(N,x)$ constructed in the proof of Theorem~\ref{localg} is appropriate. 
Indeed, it uses the technique of augmenting paths, which produces an integer flow if the edge capacities are integers.
Furthermore, in the proof of Theorem~\ref{localg}, we have shown that the flow obtained is nearly optimal in expectation. 
\end{proof}

Hence, a nearly maximum integer flow (provided an integer cost function, of course) can be approximated in $\RLA$. 
We show that it cannot be approximated by an algorithm in $\MLA$. 

\begin{center}
\includegraphics[scale=0.5]{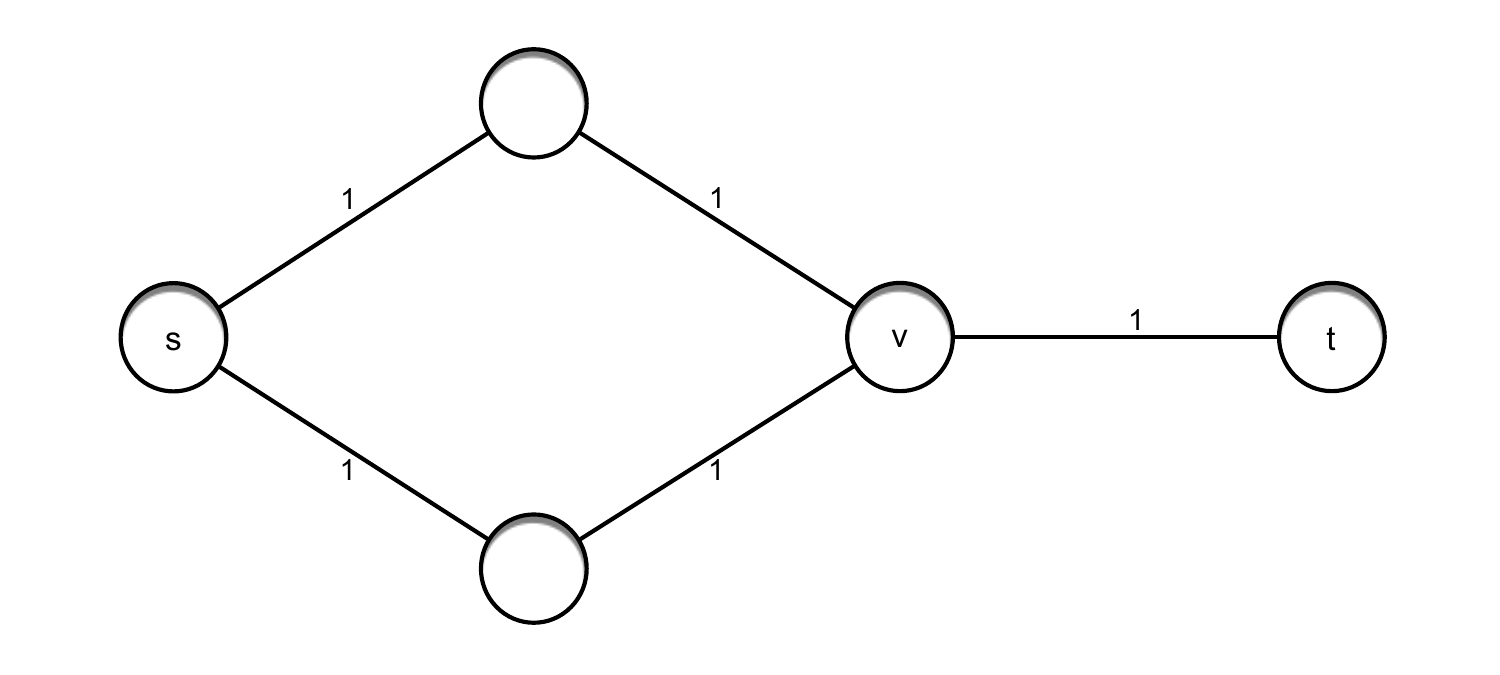}
\end{center}

Let $N$ be a network with one source $s$ and one target $t$. 
The target is only linked to a regular vertex $v$, and the source is connected to $v$ by two disjoint paths of length 2. 
If all edge capacities are 1, then the minimum cut is $\{t\}$ with value 1. 
There is of course a maximum flow with value 1 in this network: all four edges directed away from the source are at $1/2$ capacity and the edge $(v,t)$ is at full capacity 1. 
However, the maximum integer flow that can be obtained by an $\MLA$ in this network is the all-zero flow. 
Indeed, the two edges incident with the source are symmetrical, hence any flow algorithm in $\MLA$ must assign the same value to them (vertices do not have private seeds to break the symmetry). 
This common value cannot be 1, as otherwise we would obtain a flow with value 2. 
Thus an $\MLA$ must assign 0 to both edges incident with the source, yielding the all-zero flow. 

\begin{proof}[Proof of Corollary~\ref{testable}]
Let $\bar{f}_2$ be the flow constructed by the local algorithm of the proof of Theorem~\ref{localg} with error bound $\eps/2$ in \eqref{Thm1eqbounded}, which therefore satisfies $\big\|\bar{f}_2\big\| \in \big[\big\|f^*\big\| - \frac{\eps}{2} n, \big\|f^*\big\|\big]$, and let $r$ be the radius used there plus $1$. 
Let 
\begin{equation} \label{gdef}
g\Big(B_r(v_1), B_r(v_2), \ldots, B_r(v_k)\Big) = \frac{1}{k} \sum_{i=1}^k \Big( I(v_i \in S) \sum_{e \in \out(v_i)} \bar{f}_2(e) \Big)
\end{equation}

where $I(v\in S)$ denotes the indicator function of the event that the vertex $v$ is in the set of sources $S$ of the network, that is, 
$I(v\in S) = 
\begin{cases}
1\text{ if } v\in S\\
0\text{ otherwise.}
\end{cases}$

As $I(v \in S) \sum_{e \in \out(v)} \bar{f}_2(e) \in [-dM, dM]$, the variance of \eqref{gdef} is at most $d^2M^2 / k$, therefore \eqref{gdef} stochastically converges to its expected value as $k$ tends to infinity. 
This expected value is
\begin{equation*}
\frac{1}{n} \sum_{v \in V(G)} \Big( I(v \in S) \sum_{e \in \out(v)} \bar{f}_2(e) \Big) = \frac{1}{n} \sum_{e \in \out(S)} \bar{f}_2(e) = \frac{1}{n} \big\|\bar{f}_2\big\| \in \left[\frac{\big\|f^*\big\|}{n} - \frac{\eps}{2},\ \frac{\big\|f^*\big\|}{n}\right].
\end{equation*}
This implies that, for large enough $k$, these $k$, $r$ and $g$ satisfy the requirements.
\end{proof}

\subsection{Proof of Proposition~\ref{localgfraccut} and Theorem~\ref{localgcut}}\label{subsec:fraccut}

Fractional cuts are easier to work with than cuts, because just like flows, they are closed under convex combinations. 
Hence, we first show how to construct a local fractional cut algorithm by using our local flow algorithm. 
Then we reduce the problem of minimum cut to minimum fractional cut. 

\begin{proof}[Proof of Proposition~\ref{localgfraccut}]
Let us execute the local flow algorithm as in Theorem~\ref{localg} with error bound $\eps_1$. Denote the resulting flow by $f$.

We define the $\eps_2$-residual graph as
\begin{equation} \label{Gfdef}
G_f(\eps_2) = \Big(S(G), R(G), T(G), \big\{e \in \EG \big| f(e) < c(e) - \eps_2 \big\}\Big).
\end{equation}

Let $\ell = \lceil 1 / \eps_2 \rceil$. 
Denote the length of the shortest path from $S$ to a vertex $v \in V$ in the graph $\big(G_f(\eps_2)\big)$ by $\dist(v)$. 
If there is no such path then $\dist(v) = \infty$. 
We define the fractional cut $\tilde{X}$ as
\begin{alignat}{3} \label{fracconstr}
   \tilde{X}(v) &=
\begin{cases}
\min\big( \dist(v) \cdot \eps_2,\ 1 \big) & \text{if } v \in S \cup R
\\ 1 & \text{if } v \in T.
\end{cases}
\end{alignat}
This is a local fractional cut algorithm, because we need to execute the local flow algorithm for each edge only in the $\ell$-neighborhood of $v$ in order to compute $\tilde{X}(v)$.

Observe that
\begin{equation*}
\big\|f\big\| \mathop{=}^{\eqref{flownorm}} \sum_{e \in \out(S)} f(e) \mathop{=}^{\eqref{regeq}} \sum_{e \in \out(S)} f(e) + \sum_{r \in R} \Big( \big(1 - \tilde{X}(r)\big) \sum_{e \in \out(r)} f(e) \Big)
\end{equation*}
\begin{equation*}
= \sum_{s \in S} \Big(1 \cdot \sum_{e \in \out(s)} f(e)\Big) + \sum_{r \in R} \Big( \big(1 - \tilde{X}(r)\big) \sum_{e \in \out(r)} f(e) \Big) + \sum_{t \in T} \Big(0 \cdot \sum_{e \in \out(t)} f(e)\Big)
\end{equation*}
\begin{equation*}
\mathop{=}^{\eqref{fraccut}} \sum_{v \in V} \Big( \big(1 - \tilde{X}(v)\big) \sum_{e \in \out(v)} f(e) \Big) = \sum_{(a, b) \in \vE} \big(1 - \tilde{X}(a)\big) f\big((a, b)\big) 
\end{equation*}
\begin{equation*}
\mathop{=}^{\eqref{antisymmetry}} \sum_{f((a, b)) > 0} \Big( \big(1 - \tilde{X}(a)\big) f\big((a, b)\big) - \big(1 - \tilde{X}(b)\big) f\big((a, b)\big) \Big) 
= \sum_{f((a, b)) > 0} f\big((a, b)\big) \big(\tilde{X}(b) - \tilde{X}(a)\big)
\end{equation*}
\begin{equation*}
= \sum_{f((a, b)) > 0} f\big((a, b)\big) \max\big(0,\ \tilde{X}(b) - \tilde{X}(a)\big) - \sum_{f((a, b)) > 0} f\big((a, b)\big) \max\big(0,\ \tilde{X}(a) - \tilde{X}(b)\big)
\end{equation*}
\begin{equation} \label{fraceq}
\mathop{=}^{\eqref{fracnorm}} \big\|\tilde{X}\big\|_{(G, f)} - \sum_{f((a, b)) > 0} f\big((a, b)\big) \max\big(0,\ \tilde{X}(a) - \tilde{X}(b)\big).
\end{equation}

Notice that the value of a fractional cut $\|\tilde{X}\|_{(G, c)}$ is additive in $c$, because
\begin{equation*}
\big\|\tilde{X}\big\|_{(G, c_1 + c_2)} \mathop{=}^{\eqref{fracnorm}}
\sum_{(a, b) \in \vE(G)} (c_1 + c_2)\big((a, b)\big) \max\big(0,\ \tilde{X}(b) - \tilde{X}(a)\big)
\end{equation*}
\begin{equation*}
= \sum_{(a, b) \in \vE(G)} c_1\big((a, b)\big) \max\big(0,\ \tilde{X}(b) - \tilde{X}(a)\big) + \sum_{(a, b) \in \vE(G)} c_2\big((a, b)\big) \max\big(0,\ \tilde{X}(b) - \tilde{X}(a)\big)
\end{equation*}
\begin{equation} \label{additivec}
\mathop{=}^{\eqref{fracnorm}} \big\|\tilde{X}\big\|_{(G, c_1)} + \big\|\tilde{X}\big\|_{(G, c_2)}.
\end{equation}

Using this observation with $c = f + (c - f)$, and using \eqref{fraceq}, we get
\begin{equation*}
\big\|\tilde{X}\big\| = \big\|\tilde{X}\big\|_{(G, c)} \mathop{=}^{\eqref{additivec}} \big\|\tilde{X}\big\|_{(G, f)} + \big\|\tilde{X}\big\|_{(G, c - f)}
\end{equation*}
\begin{equation} \label{threeterms}
\mathop{=}^{\eqref{fraceq}} \big\|f\big\| + \sum_{f((a, b)) > 0} f\big((a, b)\big) \max\big(0,\ \tilde{X}(a) - \tilde{X}(b)\big) + \big\|\tilde{X}\big\|_{(G, c - f)}.
\end{equation}

We call an edge $(a, b) \in \vE(N)$ \textbf{bad} if there exists a short augmenting path in $G_f(\eps_2)$ finishing with the edge $(a, b)$. 
Or equivalently, $(a, b) \in \vE(N)$ is bad if $b \in T$ and there exists a path in $G_f(\eps_2)$ from $S$ to $a$ of length at most $\ell - 1$. 
Consider the terms $f\big((a, b)\big) \max \big( 0,\ \bar{X}(a) - \bar{X}(b) \big)$. If $f\big((a, b)\big) \ge \eps_2$, then $(b, a) \in \vE\big(G_f(\eps_2)\big)$ because $f\big((b, a)\big) = - f\big((a, b)\big) < - \eps_2 \le c(b, a) - \eps_2$.
By $\eqref{fracconstr}$, unless if $(a, b)$ is bad, this implies $\tilde{X}(a) - \tilde{X}(b) \le \eps_2$. 
Therefore, \underline{if $(a, b)$ is not bad}, then either $f\big((a, b)\big)$ or $\max \big( 0, \bar{X}(a) - \bar{X}(b) \big)$ is at most $\eps_2$, and both are $\in [0, 1]$, thus
\begin{equation} \label{term2a}
f\big((a, b)\big) \max \big( 0,\ \bar{X}(a) - \bar{X}(b) \big) \le \eps_2.
\end{equation}

As $\|f\| \ge \|f^*\| - \eps_1 n$, there exist at most $\frac{\eps_1}{\eps_2} n$ edge-disjoint augmenting paths with capacity at least $\eps_2$. 
Each short augmenting path intersect at most $z = z(\ell)$ number of short augmenting paths, which yields that there are at most $\frac{\eps_1 z}{\eps_2} n$ number of short augmenting paths with capacity at least $\eps_2$. 
Therefore, there exist at most $\frac{\eps_1 z}{\eps_2} n$ number of bad edges. 
Combining this estimation with \eqref{term2a} and with the fact that the graph has at most $dn / 2$ number of edges, we obtain
\begin{equation} \label{term2}
\sum_{f((a, b)) > 0} f\big((a, b)\big) \max\big(0,\ \tilde{X}(a) - \tilde{X}(b)\big) \le \frac{\eps_1 z}{\eps_2} n + \frac{\eps_2 d}{2} n
\end{equation}

Consider now the terms of the sum
\begin{equation} \label{term3a}
\big\| \tilde{X} \big\|_{(G, c-f)} \mathop{=}^{\eqref{fracnorm}} \sum_{(a, b) \in \vE(N)} (c - f)\big((a, b)\big) \max\big(0,\ \tilde{X}(b) - \tilde{X}(a)\big)
\end{equation}
Assume that $(a, b) \in \vE(N)$ is not bad. If $(c-f)\big((a, b)\big) > \eps_2$, then \eqref{Gfdef} shows that $(a, b) \in \vE\big(G_f(\eps_2)\big)$, therefore \eqref{fracconstr} implies $\tilde{X}(b) - \tilde{X}(a) \le \eps_2$. Hence, similarly to \eqref{term2a}, we get that $(c - f)\big((a, b)\big) \max\big(0,\ \tilde{X}(b) - \tilde{X}(a)\big) \le \eps_2$, therefore similarly to \eqref{term2}, we obtain
\begin{equation} \label{term3}
\sum_{(a, b) \in \vE(N)} (c - f)\big((a, b)\big) \max\big(0,\ \tilde{X}(b) - \tilde{X}(a)\big) \le \frac{\eps_1 z}{\eps_2} n + \frac{\eps_2 d}{2} n.
\end{equation}

By putting $\eps_1 = \eps^2 / (8dz)$ and $\eps_2 = \eps / (2d)$, we have $\eps_1/\eps_2 = \eps/(4z)$, and thus 
\begin{equation*}
\big\| \tilde{X} \big\| \mathop{=}^{\eqref{fracnorm}} \big\| f \big\| + \sum_{f((a, b)) > 0} f\big((a, b)\big) \max\big(0,\ \tilde{X}(a) - \tilde{X}(b)\big) + \big\| \tilde{X} \big\|_{(G, c - f)}
\end{equation*}
\begin{equation*}
\mathop{\le}^{(\ref{term2},\ref{term3})} \big\| f^* \big\| + 2 \Big( \frac{\eps_1 z}{\eps_2} n + \frac{\eps_2 d}{2} n \Big) = \big\| f^* \big\| + 2 \big( \frac{\eps}{4z}zn + \frac{\eps}{2d}\frac{d}{2}n \big) = \big\| f^* \big\| + \eps n. \qedhere
\end{equation*}
\end{proof}

\begin{proof}[Proof of Theorem~\ref{localgcut}]
Consider the local fractional cut algorithm $\tilde{C}$ as in Proposition~\ref{localgfraccut}. 
Then for each network $N$, $\big\| \tilde{C}(N) \big\| \le \big\|f^*\big\| + \eps n$. 
For all $u \in (0, 1)$, $\tilde{X}[u]$ is a local cut algorithm, and with a uniform random $u \in (0, 1)$, equation \eqref{fractocut} shows that we get a probability distribution of local cut algorithms producing the desired expected value.
\end{proof}

\subsection{Proof of Theorem~\ref{separation} and Corollary~\ref{smallgraph}}\label{subs:pumping}

\begin{Lemma} \label{convex}
For any nice family $\F$, $S_r(\F)$ is a convex compact subset of $\R^{B_r}$.
\end{Lemma}

\begin{proof}
For any two graphs $G_1, G_2 \in \F$ and $k_1,k_2 \in \N$,
\begin{equation*}
s_r\big((k_1 \times G_1) \cup (k_2 \times G_2)\big) = \frac{k_1 |V(G_1)| \cdot s_r(G_1) + k_2 |V(G_2)| \cdot s_r(G_2)}{k_1 |V(G_1)| + k_2 |V(G_2)|},
\end{equation*}
where $k \times G$ denotes the disjoint union of $k$ isomorphic copies of $G$. 
This implies that each convex combination of the $r$-neighborhood statistics of two graphs in $\F$ can be approximated by the $r$-neighborhood statistics of another graph in $\F$.
\end{proof}

Let us identify the natural base of $\R^{B_r}$ by the elements of $B_r$, and let the linear extension of $w\: B_r \rightarrow \R$ be defined as the function $\tilde{w} \: \R^{B_r} \rightarrow \R$,
\begin{equation} \label{linex}
\tilde{w}\Big( \sum_{b \in B_r} \lambda_b b \Big) = \sum_{b \in B_r} \lambda_b w(b).
\end{equation}
Let us define
\begin{equation} \label{mdef}
m(\F, w) = \max_{P \in S_r(\F)} \sum_{b \in B_r} w(b) P(b) \mathop{=}^{\eqref{linex}} \max_{P \in S_r(\F)} \tilde{w}(P) \mathop{=}^{\eqref{Ddef}} \sup_{G \in \F} \tilde{w}\big( s_r(G) \big).
\end{equation}

\begin{Lemma}\label{lem:convexdist}
For each distribution $P$ on $B_r$,
\begin{equation*}
P \in S_r(\F)\ \ \Leftrightarrow\ \ \forall w\: B_r \rightarrow [0, 1],\ \tilde{w}(P) \le m(\F, w).
\end{equation*}
\end{Lemma}

\begin{proof}
The set $S_r(\F)$ is convex and compact, therefore, it is determined by its dual:
\begin{equation*}
P \in S_r(\F)\ \ \Leftrightarrow\ \ \forall w\: B_r \rightarrow \R,\ \tilde{w}(P) \le m(\F, w).
\end{equation*}
Furthermore, for each $\lambda > 0$ and $c \in \R$, we have $m(\F, \lambda w + c) = \lambda m(\F, w) + c$. This shows that if we know $m(\F, w)$ for all $w\: B_r \rightarrow [0, 1]$, then we know it for all $w\: B_r \rightarrow \R$.
\end{proof}

We note without proof that the $L_1$-distance of a distribution of $r$-neighborhoods $P \in \R^{B_r}$ from the convex compact set $S_r(\F)$ is 

\begin{equation}\label{eq:minmaxdist}
\min\limits_{Q\in S_r(\F)} \big\| P-Q \big\|_1= \max(0, \max\limits_{w\: B_r \rightarrow [0, 1]} \tilde{w}(P) - m(\F, w)).
\end{equation}

We state the generalization of Theorem~\ref{separation} as indicated in Subsection~\ref{subsec:AL}, and then prove it after a series of lemmas. 

\begin{Theorem} \label{separation2}
Assume that $s_r(G_0) \notin S_r(\F)$ holds for a nice family $\F$ of graphs, a graph $G_0 \in \G$, and $r\in \N$. 
Namely, there exists a $w_0\: B_r \rightarrow [0, 1]$ satisfying
\begin{equation} \label{deltadef}
\tilde{w}_0\big(s_r(G_0)\big) - m(\F, w_0) = \delta > 0.
\end{equation}
Then for all $\eps > 0$, there is an $r' = r'(r, \eps, \delta)$, a subset $B' \subset B_{r'}$ and a spanned subgraph $G_1$ of $G_0$ such that 
\begin{align*}
\P\big(s_{r'}(G_1) \in B'\big) &> 1 -\eps, \,\,and\\
\forall G \in \F\: \,\, \P\big(s_{r'}(G) \in B'\big) &< \eps. 
\end{align*}
\end{Theorem}

Before the proof, we first need to introduce some notions. 
If for a graph $G \in \G$ and a function $w \: B_r \rightarrow [0, 1]$ we have 
\begin{equation} \label{supeq}
\tilde{w}\big(s_r(G)\big) = \sup_{G' \subseteq G} \tilde{w}\big(s_r(G')\big),
\end{equation}
then we say that $G$ is \textbf{supremal} for $w$.

For a graph $G \in \G$, radius $r \in \N$, weighting $w\: B_r \rightarrow [0, 1]$, and $\alpha > 0$, we define the averaging network $A = A(G, r, w, \alpha) = \big((S, R, T, \vE), c\big)$. The sets $S$, $R$, $T$ are three copies of $V(G)$, and for each $v \in V(G)$, we denote its copies by $v_S$, $v$ and $v_T$, respectively. We identify $R$ with $V(G)$. Let $\vE = \vE(G) \cup \big\{ (v_S, v), (v, v_S), (v, v_T), (v_T, v)$ $\big| v \in V(G) \big\}$. For each $e \in \vE(G)$, let $c(e) = d ^ r$, and for each vertex $v \in V(G)$, let $c\big((v_S, v)\big) = w\big(B_r(G, v)\big)$ and $c\big((v, v_T)\big) = \alpha$, and $c\big((v, v_S)\big) = c\big((v_T, v)\big) = 0$.

\begin{Lemma} \label{aflowsize}
Given a graph $G \in \G$, a radius $r \in \N$, a function $w \: B_r \rightarrow [0, 1]$, and $\alpha > 0$, the size of the maximum flow $\big\|f^*\big\|$ in $A(G, r, w, \alpha)$ satisfies
\begin{equation}\label{flowsize}
\min\Big(\alpha,\ \inf_{G' \subseteq G} \tilde{w}\big(s_r(G')\big)\Big) n
\le \big\|f^*\big\|
\le^{(*)} \min\Big(\alpha,\ \tilde{w}\big(s_r(G)\big)\Big) n
\end{equation}
with equality at (*) if $G$ is supremal for $w$.
\end{Lemma}

\begin{proof}
Notice that for any graph $G \in \G$,
\begin{equation}\label{tildew}
\tilde{w}\big(s_r(G)\big) = \frac{1}{|V(G)|} \sum_{v \in V(G)} w\big(B_r(G, v)\big).
\end{equation}

The upper bound in \eqref{flowsize} follows by computing the value of two cuts: $\big\|S \cup R\big\| = \alpha n$ and $\big\|S\big\| = \tilde{w}\big(s_r(G)\big) n$. 

On the other hand, consider an arbitrary cut $X$ of $A$. 
Let $R^- = R \cap X$ and $R^+ = R \backslash X$. 
Let $G^-$ and $G^+$ denote the subgraph of $G$ spanned by $R^-$ and $R^+$, respectively. 
Let $\delta_X$ denote the number of edges between $R^-$ and $R^+$. 
For each edge between $R^-$ and $R^+$, its $r$-neighborhood contains at most $d^r$ number of vertices. 
Therefore there exist at most $d^r \delta_X$ number of vertices $v \in R^-$ for which $B_r(G^-, v) \ne B_r(G, v)$, and at most $d^r \delta_X$ number of vertices $v \in R^+$ for which $B_r(G^+, v) \ne B_r(G, v)$. 
These imply the following inequalities.
\begin{equation} \label{drup}
\sum_{v \in R^-} \Big( w\big(B_r(G, w)\big) - w\big(B_r(G^-, w)\big) \Big) \le d^r \delta_X.
\end{equation}
\begin{equation} \label{drlow}
\sum_{v \in R^+} \Big( w\big(B_r(G^+, w)\big) - w\big(B_r(G, w)\big) \Big) \le d^r \delta_X.
\end{equation}
Now we prove the lower bound of \eqref{flowsize}.
\begin{equation*}
\big\|X\big\| \mathop{=}^{\eqref{cutnorm}} \big|R^-\big| \alpha + \sum_{v \in R^+} w\big(B_r(G, v)\big) +  d ^ r \delta_X \mathop{\ge}^{\eqref{drlow}} \big|R^-\big| \alpha + \sum_{v \in R^+} w\big(B_r(G^+, v)\big)
\end{equation*}
\begin{equation*}
\mathop{=}^{\eqref{tildew}} \big|R^-\big| \alpha + \big|R^+\big| \tilde{w}\big(s_r(G^+)\big) \ge \Big(\big|R^-\big| + \big|R^+\big|\Big) \min\Big(\alpha,\ \tilde{w}\big(s_r(G^+)\big)\Big)
\end{equation*}
\begin{equation*}
= \min\Big(\alpha,\  \tilde{w}\big(s_r(G^+)\big)\Big) n \ge
\min\Big(\alpha,\ \inf_{G' \subseteq G} \tilde{w}\big(s_r(G')\big)\Big) n.
\end{equation*}
Finally, we show that if $G$ is supremal for $w$, then the upper bound of \eqref{flowsize} is also a lower bound.
\begin{equation*}
\big\|X\big\| \mathop{=}^{\eqref{cutnorm}} \big|R^-\big| \alpha + \sum_{v \in R^+} w\big(B_r(G, v)\big) +  d ^ r \delta_X
= \big|R^-\big| \alpha + \sum_{v \in R} w\big(B_r(G, v)\big) - \sum_{v \in R^-} w\big(B_r(G, v)\big)
\end{equation*}
\begin{equation*}
+ d ^ r \delta_X \mathop{\ge}^{\eqref{drup}} \big|R^-\big| \alpha + \sum_{v \in R} w\big(B_r(G, v)\big) - \sum_{v \in R^-} w\big(B_r(G^-, v)\big) \mathop{=}^{\eqref{tildew}} \big|R^-\big| \alpha + \big|R\big| \tilde{w}\big(s_r(G)\big)
\end{equation*}
\begin{equation*}
- \big|R^-\big| \tilde{w}\big(s_r(G^-)\big) = \big|R^-\big| \alpha + \big|R^+\big| \tilde{w}\big(s_r(G)\big) + \big|R^-\big| \Big( \tilde{w}\big(s_r(G)\big) - \tilde{w}\big(s_r(G^-)\big) \Big) \mathop{\ge}^{\eqref{supeq}}
\end{equation*}
\begin{equation*}
\big|R^-\big| \alpha + \big|R^+\big| \tilde{w}\big(s_r(G)\big)
\ge \Big(\big|R^-\big| + \big|R^+\big|\Big) \min\Big(\alpha,\ \tilde{w}\big(s_r(G)\big)\Big)
\ge \min\Big(\alpha,\ \tilde{w}\big(s_r(G)\big)\Big) n. \qedhere
\end{equation*}
\end{proof}

For a radius $r$, weighting $w \: B_r \rightarrow [0, 1]$, $\eps > 0$ and $\alpha > 0$, we define the following operator $W_{\eps}(r, w, \alpha)$. Its value will be a new weighting $w' \: B_{r'} \rightarrow [0, \alpha]$, where $r' = g_1(r, \eps)$ is $r$ plus the radius used in Theorem~\ref{localg} with error bound $\eps / d^r$.

Consider an arbitrary $B \in B_{r'}$ with root $v$. Let $A_0 = A(B, r, w, \alpha)$. Consider the flow $f$ generated by the local flow algorithm in Theorem~\ref{localg} on $A_0$ with error bound $\eps$. 
\begin{equation} \label{wprimedef}
w'(B) = \alpha - f\big((v, v_T)\big).
\end{equation}
Notice that
\begin{equation} \label{wsupport}
w'(B) = \alpha - f\big((v, v_T)\big) \in \alpha - \Big[-c\big((v_T, v)\big), c\big((v, v_T)\big)\Big] = \alpha - [0, \alpha] = [0, \alpha].
\end{equation}

\begin{Lemma} \label{WHG}
For each graph $G$, radius $r$, weighting $w \: B_r \rightarrow [0, 1]$, $\eps > 0$, $\alpha > 0$, $r' = g_1(r, \eps)$ and $w' = W_{\eps}(r, w, \alpha)$,
\begin{equation} \label{lower}
\max\Big(0,\ \alpha - \tilde{w}\big(s_r(G)\big)\Big) \le \tilde{w}'\big(s_{r'}(G)\big) \le \max\Big(0,\ \alpha - \inf_{G' \subseteq G} \tilde{w}\big(s_r(G')\big)\Big) + \eps.
\end{equation}
Furthermore, if $G$ is supremal for $w$, then the lower bound is $\eps$-tight, namely
\begin{equation} \label{epsbound}
\tilde{w}'\big(s_{r'}(G)\big) \le \max\Big(0,\ \alpha - \tilde{w}\big(s_r(G)\big)\Big) + \eps.
\end{equation}
\end{Lemma}

\begin{proof}
Let $A_1 = A(G, r, w, \alpha)$, and consider the flow $f$ generated by the local flow algorithm in Theorem~\ref{localg} on $A_1$ with error bound $\eps$. For each $v \in V(G)$, $f\big((v, v_T)\big)$ depends only on $B_{r'}(G, v)$, therefore \eqref{wprimedef} shows that
\begin{equation} \label{wprime}
w'\big(B_{r'}(G, v)\big) = \alpha - f\big((v, v_T)\big).
\end{equation}
Hence, we obtain
\begin{equation*}
\tilde{w}'\big(s_{r'}(G)\big)
\mathop{=}^{\eqref{tildew}} \frac{1}{n} \sum_{v \in V(G)} w'\big(B_r(G, v)\big)
\mathop{=}^{\eqref{wprime}} \frac{1}{n} \sum_{v \in V(G)} \Big(\alpha - f\big((v, v_T)\big)\Big)
\end{equation*}
\begin{equation*}
= \alpha - \frac{1}{n} \sum_{v \in V(G)} f\big((v, v_T)\big)
= \alpha - \frac{1}{n}\big\|f\big\|
\mathop{\in}^{\eqref{Thm1eq}} \alpha - \frac{1}{n}\Big[\big\|f^*\big\| - \eps n, \big\|f^*\big\|\Big] 
\end{equation*}
\begin{equation*}
\mathop{\subseteq}^{\eqref{flowsize}} \alpha - \bigg[ \min\Big(\alpha,\ \inf_{G' \subseteq G} \tilde{w}\big(s_r(G')\big)\Big) - \eps,\ \min\Big(\alpha,\ \tilde{w}\big(s_r(G)\big)\Big) \bigg]
\end{equation*}
\begin{equation*}
= \bigg[\max\Big(0,\ \alpha - \tilde{w}\big(s_r(G)\big)\Big),\ \max\Big(0,\ \alpha - \inf_{G' \subseteq G} \tilde{w}\big(s_r(G')\big)\Big) + \eps \bigg] \qedhere
\end{equation*}
\end{proof}

We are ready to prove Theorem~\ref{separation2} and then Corollary~\ref{smallgraph}. 

\begin{proof}[Proof of Theorem~\ref{separation2}]
Let us set
\begin{equation} \label{eps1}
\eps_1 = \delta \eps / 2,
\end{equation}
and let $w_1 = W_{\eps_1}\big(r, w_0, \tilde{w}_0(s_r(G_0))\big)$ with radius $r_1$. Then, let $w_2 = W_{\eps_1}\big(r_1, w_1, \delta\big)$, and let $r'$ be the radius it uses. Let us define
\begin{equation} \label{Mdef}
B' = \big\{ B \in B_{r'}: w_2(B) > \delta / 2 \big\}.
\end{equation}
Let $G_1$ be the spanned subgraph of $G_0$ with the lowest $\tilde{w}_1\big(s_{r_1}(G_1)\big)$. We show that this satisfies the requirements.

Let $G_2$ be the spanned subgraph of $G_0$ with the highest $\tilde{w}_0\big(s_r(G_2)\big)$. These also imply $G_1$ is infimal for $w_1$, and $G_2$ is supremal for $w_0$, namely
\begin{align}
\forall G' \subseteq G_0:&& \tilde{w}_1\big(s_{r_1}(G')\big) \ge \tilde{w}_1\big(s_{r_1}(G_1)\big) &= \inf_{G \subseteq G_1} \tilde{w}_1\big(s_{r_1}(G)\big), \label{G1inf}
\\ \forall G' \subseteq G_0:&& \tilde{w}_0\big(s_r(G')\big) \le \tilde{w}_0\big(s_r(G_2)\big) &= \sup_{G \subseteq G_2} \tilde{w}_0\big(s_r(G)\big). \label{G2sup}
\end{align}

\begin{equation*}
\tilde{w}_2\big(s_{r'}(G_1)\big)
\mathop{\ge}^{\eqref{lower}}
\delta - \tilde{w}_1\big(s_{r_1}(G_1)\big)
\mathop{\ge}^{\eqref{G1inf}} \delta - \tilde{w}_1\big(s_{r_1}(G_2)\big)
\end{equation*}
\begin{equation} \label{w2G1}
\mathop{\ge}^{\eqref{epsbound}}
\delta - \bigg( \max\Big(0,\ \tilde{w}_0\big(s_r(G_0)\big) - \tilde{w}_0\big(s_r(G_2)\big)\Big) + \eps_1 \bigg)
\mathop{\ge}^{\eqref{G2sup}}
\delta - \eps_1.
\end{equation}
\begin{equation} \label{PG1}
\P\big(s_{r'}(G_1) \in B'\big)
\mathop{=}^{\eqref{Mdef}} \P\Big( w_2\big(s_{r'}(G_1)\big) > \delta / 2 \Big)
= \P\Big( \delta - w_2\big(s_{r'}(G_1)\big) \le \delta / 2 \Big)
\end{equation}

Identity \eqref{wsupport} shows that $\forall B \in B_{r'}\: w_2(B) \le \delta$, therefore using Markov's inequality for $\delta - w_2\big(s_{r'}(G_1)\big) > 0$,

\begin{equation*}
\eqref{PG1} \ge 1 - \frac{\E\big(\delta - w_2(s_{r'}(G_1))\big)}{\delta / 2}
=  1 - \frac{\delta - \E\big(w_2(s_{r'}(G_1))\big)}{\delta / 2}
\end{equation*}
\begin{equation*}
\mathop{=}^{\eqref{linex}} 1 - \frac{\delta - \tilde{w}_2\big(s_{r'}(G_1)\big)}{\delta / 2}
\mathop{\ge}^{\eqref{w2G1}} 1 - \frac{\delta - (\delta -\eps_1)}{\delta / 2}
= 1 - \frac{2 \eps_1}{\delta} \mathop{=}^{\eqref{eps1}} 1 - \eps.
\end{equation*}

On the other hand, for all $G \in \F$,
\begin{equation*}
\tilde{w}_1\big(s_{r_1}(G)\big)
\mathop{\ge}^{\eqref{lower}}
\tilde{w}_0\big(s_r(G_0)\big) - \tilde{w}_0\big(s_r(G)\big)
\mathop{\ge}^{\eqref{mdef}} \tilde{w}_0\big(s_r(G_0)\big) - m(\F, w_0)
\mathop{=}^{\eqref{deltadef}} \delta.
\end{equation*}
Therefore, as $\F$ is closed under taking spanned subgraphs, for all $G \in \F$ we have 
\begin{equation} \label{spandelta}
\inf_{G \subseteq G_1} \tilde{w_1}\big(s_{r_1}(G)\big) \ge \delta.
\end{equation}

Therefore for all $G \in \F$,
\begin{equation} \label{wG}
\tilde{w}_2\big(s_{r'}(G)\big)
\mathop{\le}^{\eqref{lower}}
\max\Big(0,\ \delta - \inf_{G \subseteq G_1} \tilde{w_1}\big(s_{r_1}(G)\big)\Big) + \eps_1 \mathop{=}^{\eqref{spandelta}} \eps_1.
\end{equation}
Using Markov's inequality for $w_2\big(s_{r'}(G)\big) > 0$,
\begin{equation*}
\P\big(s_{r'}(G) \in B'\big)
\mathop{=}^{\eqref{PG1}} \P\Big(w_2\big(s_{r'}(G)\big) > \delta / 2 \Big)
\le \frac{\E\big(w_2(s_{r'}(G))\big)}{\delta / 2}
\mathop{=}^{\eqref{linex}} \frac{\tilde{w}_2\big(s_{r'}(G)\big)}{\delta / 2}
\mathop{\le}^{\eqref{wG}} \frac{2 \eps_1}{\delta} \mathop{=}^{\eqref{eps1}} \eps. \qedhere
\end{equation*}
\end{proof}

\begin{proof}[Proof of Corollary~\ref{smallgraph}]
Let $\F$ be the nice family of all graphs with degrees bounded by $d$, all of whose components have size at most $n(r, \delta) - 1$. 
By the definition of $n(r, \delta)$, there exists a graph $G_0 \in \G$ such that $\big\| s_r(G_0)-s_r(G)\big\|\geq \delta$ for all $G\in \F$. 
Therefore, by equation~\eqref{eq:minmaxdist}, there exists a $w_0\: B_r \rightarrow [0, 1]$ satisfying $\tilde{w}_0\big(s_r(G_0)\big) - m(\F, w_0) \ge \delta$. 
Then applying Theorem~\ref{separation2} with these $\F$, $r$, $\delta$, $w_0$, and $\eps = \frac{1 - \lambda}{2}$ provides us with $r'$ and $B' \subset B_{r'}$. 
Let $w_1$ be the characteristic function of $B'$, namely, $w_1: B_{r'} \rightarrow \{0, 1\}$ such that $w_1(b) = 1$ if and only if $b \in B'$. 
By Theorem~\ref{separation2} we have

$$\exists G_1 \in \G \: \forall G \in \F \: \,\,\,\, \tilde{w}_1\big(s_r(G_1)\big) - \tilde{w}_1\big(s_r(G)\big) > (1 - \eps) - \eps = \lambda.$$ 

Consequently, $n(r, \delta) - 1$ does not satisfy the modified conditions of \eqref{ndef} for radius $r'$ and error bound $\lambda$. 
Thus $n(r, \delta) \le n\big(r', \lambda\big)$.
\end{proof}

\section{Unimodular random graphs}

Unimodular random graphs (URGs) are probability distributions on rooted connected graphs with finite degrees satisfying the unimodularity condition. 
Informally, this condition expresses that a URG behaves like a finite graph rooted at a uniform random vertex. 
In particular, every finite connected graph with uniformly chosen root is a URG. 
We are only focusing on URGs with given maximum degree $d$. 

For the formal definition, let $\mathfrak{G}^\bullet$ be the set of all rooted, connected countable graphs with degrees bounded by $d$. 
In order to define the natural topology on $\mathfrak{G}^\bullet$, we endow it with a metric: two rooted graphs with centers $o$ and $o'$ are at distance $2^{-r}$, where $r$ is the smallest radius such that $B(o, r)\not\cong B(o',r)$. 
This makes $\mathfrak{G}^\bullet$ a completely disconnected compact metric space. 

Let the measure $\sigma$ be a Borel measure defined on the induced Borel sigma-algebra $\mathfrak{A}$. 
We define the probability measure $\sigma^*$ by $\sigma^*(A)=\int\limits_A \deg d\sigma / \int\limits_{\mathfrak{G}^\bullet} \deg d\sigma$. 
Select a connected rooted graph according to the measure $\sigma^*$ and then select a uniform random edge $e$ incident with the root. 
This way we get a probability measure on rooted connected graphs (all of whose vertices have degree at most $d$) with a ``root edge'' incident to the root. 
If this measure is invariant under moving the root to the other endpoint of the root edge, then we say that $\sigma$ is involution invariant, or alternatively, that the measure is a \textbf{unimodular} random graph, or $\URG$ for short. 
Further details, examples and alternative reformulations, including the Mass Transport Principle can be found in \cite{Lovaszbook}. 
Note that a finite connected graph with maximum degree $d$ can always be viewed as a $\URG$ by setting $\sigma$ as the uniform distribution on the vertices (inducing a distribution $\sigma^*$ on rooted graphs where the probability of picking a vertex as a root is proportionate to its degree, and in return, a uniform distribution on the root edges).

We can refine the Borel sigma algebra by decorating the vertices and edges of graphs in $\mathfrak{G}^\bullet$, and repeat the above construction. 
To define unimodular random flow networks ($\URN$s), we decorate each vertex by exactly one of the letters $S$, $R$, $T$, and each directed edge by a non-negative real number, the value of a cost function. 
We assume that this cost function is bounded by a given $M>0$, that is $c(e) \in [0,M]$ for each edge $e$. 
As the set of decorating labels $\{S,R,T\}$ of vertices and $[0,M]$ of edges are compact, the induced topological space is once again compact. 
A probability measure $\sigma$ defines a $\URN$ if 
\begin{enumerate}
\item $\sigma$ is a Borel probability measure on the larger Borel sigma algebra $\mathfrak{A}'$ obtained this way (where we are allowed to use expressions such as $v\in S$ or $c(e)<M/2$ when defining Borel sets), and 
\item the probability measure $\sigma^*$ induced on the rooted networks with a root edge as above is invariant under moving the root to the other endpoint of the root edge. 
\end{enumerate}
In particular, the restriction $\sigma\upharpoonright_\mathfrak{A}$ to the original Borel sigma algebra is a unimodular random graph. 
A unimodular flow of a $\URN$ is a flow that is a Borel function with respect to $\mathfrak{A}'$. 
The Bernoulli graphing is obtained by decorating the vertices (and/or edges, and in fact possibly all short paths) of graphs $v\in \mathfrak{G}^\bullet$ by a uniform random seed $x(v)\in [0,1]$ (and/or $x(e)\in [0,1]$, and/or $x(u)\in [0,1]$). 
Given a $\URG$ $\sigma$, the measure $\sigma'$ on the above topological space can be defined by first sampling a graph from $\sigma$, and then assigning a uniform seed in $[0,1]$ to all vertices independently. 
The above two refinements can be combined: we may decorate the vertices/edges/short paths of a $\URN$ by uniform random seeds in $[0,1]$. 
Once again, we obtain a compact space for the same reason as before. 
Note that a construction using these random seeds might not be Borel in the sense of the $\URN$, as the network is oblivious to the seeds. 
However, taking the conditional expectation with respect to the product measure of the seeds (if possible) might lead to a construction that is legal in the network sense. 
A simplified form of this idea has already been used in the paper: the concept of a chain involved random seeds $x(u)$ corresponding to short paths, and when applying Lemma~\ref{spread} in the proof of Theorem~\ref{localg}, we took the expectation over such seeds. 
Fortunately, the expected value of random flows is a flow (in this case, automatically a unimodular flow), and the expected value of random fractional cuts is a fractional cut. 
The same assertion fails for cuts in general: expectation of cuts yields a fractional cut, which is one of the reasons we work with fractional cuts. 

A \textbf{factor fractional spanned subgraph} of the $\URG$ $\sigma$ (or a fractional cut of the $\URN$ $\sigma$) is a measurable function from $(\mathfrak{G}^\bullet, \mathfrak{A}, \sigma)$ to $[0,1]$ (equipped with the Lebesgue-measure). 
It is a \textbf{factor spanned subgraph} of the $\URG$ $\sigma$ (or a cut of the $\URN$ $\sigma$) if it assigns 0 or 1 to $\sigma$-almost all rooted graphs. 
We can think about the subgraph as those rooted graphs where the function attains the value 1. 
A \textbf{factor of i.i.d.\ ($\FIID$) fractional spanned subgraph} of the $\URG$ $\sigma$ (or a fractional cut of the $\URN$ $\sigma$) is a measurable function from the Bernoulli graphing to $[0,1]$. 
More precisely, such a measurable function induces a unimodular fractional spanned subgraph of $\sigma$. 
It is an \textbf{$\FIID$ spanned subgraph} of the $\URG$ $\sigma$ (or a cut of the $\URN$ $\sigma$) if the measurable function has range in $\{0,1\}$ (almost everywhere). 
If we want to emphasize that we are talking about a subgraph rather than a fractional subgraph, or a cut rather than a fractional cut, we can call these objects integer subgraphs or integer cuts. 
These give rise to special unimodular spanned (fractional or integer) subgraphs of $\sigma$, and in particular, are themselves $\URG$s: they induce an involution invariant measure $\tau$ on $(\mathfrak{G}^\bullet, \mathfrak{A})$. 
The fact that $\tau$ is an involution invariant measure obtained this way from an $\FIID$ spanned subgraph of $\sigma$ is indicated by the notation $\tau\subseteq \sigma$.  
The notion of a factor of i.i.d.\ process is a central concept in local algorithms. 
It is relevant to us because the local flow algorithm presented in the proof of Theorem~\ref{localg} was constructed as the expectation of a factor of i.i.d.\ process with respect to the seeds. 
Most results in this paper hinges on that factor of i.i.d.\ process, and it is also used in the generalized results to come.  
We provide further motivation to consider $\FIID$ spanned (fractional) subgraphs rather than the weaker notion of unimodular random (fractional) subgraphs after the discussion on the Aldous-Lyons Conjecture.


\subsection{Generalization of the main results to $\URG$s and $\URN$s}


First of all, we generalize the main theorems to unimodular random graphs and networks. 
Note that unimodular random graphs do not admit the notion of the number of vertices, and similarly, we cannot define values of flows and cuts in $\URG$s. 
Due to such obstructions, we need to find equivalent reformulations of the theorems that can be generalized to $\URG$s and $\URN$s. 

Following standard terminology, we introduce the relative size of the values of flows and cuts in proportion to the number of vertices: see for instance the independence ratio, which is the ratio of the size of the largest independent set and the size of the vertex set. 
Alternatively, one could divide by the size of the edge set rather than the vertex set. 
The conversion rate between the two alternative ratios is the average degree, which is easily accessible from $\sigma$. 

Therefore, we define the flow ratio of a finite network $N$ as $\big\|\underline{f}^*(N)\big\| = \big\|f^*(N)\big\|/n$, the relative flow value of a flow $f$ as $\big\|\underline{f}(N)\big\| =\big\| f(N) \big\|/n$, the relative fractional cut value of a fractional cut $\tilde{X}$ as $\big\|\underline{\tilde{X}}(N)\big\| =\big\| \tilde{X}(N) \big\|/n$, and the relative capacity as $\underline{\capp}(N) = \capp(N)/n$. 
Then Theorem~\ref{localg} claims the existence of a local flow algorithm for any $\eps>0$ whose output is a flow with relative flow value $\eps$ close to the flow ratio. 
Similarly, Theorem~\ref{localgcut} states that assuming a universal bound on the capacities, there is a probability measure of local cut algorithms such that the output of a random sample returns a cut with expected relative cut value $\eps$ close to the flow ratio. 

All the above notions have a natural extension to $\URN$s. 
Let $\sigma$ be an involution invariant measure on $(\mathfrak{G}^\bullet, \mathfrak{A})$. 
We define the random variable $Y$ as follows. 
If the root $o$ is in the set of sources $S$, then $Y=\sum\limits_{(o,o')\in\vE} f((o,o'))$; otherwise $Y=0$. 
Observe that $\big\|\underline{f}(N)\big\| = \mathbb{E}_\sigma (Y) = \mathbb{P}_\sigma(o\in S)\mathbb{E}_\sigma \bigg(\sum\limits_{(o,o')\in\vE} f((o,o')) \mid o\in S\bigg)$ for finite networks. 
Hence, we define $\big\|\underline{f}(N)\big\|$ as $\mathbb{E}_\sigma (Y)$ for $\URN$s. 
If $\tilde{X}$ is a fractional cut and $e=(o,o')$ is the root edge, then the random variable $Y=c\big((o, o')\big) \max\big(0,\ \tilde{X}(o') - \tilde{X}(o)\big)$ has mean $\mathbb{E}_\sigma (Y) = \big\|\underline{\tilde{X}}(N)\big\|$ for finite networks, thus this is the plausible way to define $\big\|\underline{\tilde{X}}(N)\big\|$ for $\URN$s. 
Finally, the relative capacity is $\underline{\capp}(N) = \mathbb{E}_\sigma (Y)$ where $Y=c((o,o'))$ is the capacity of the root edge. 
Note that maximum flows $f^*$, minimum cuts and minimum fractional cuts exist for $\URN$s by compactness: at this point however, we do not know whether their values coincide. 
Of course, the minimum cut value and the minimum fractional cut value must be equal due to the same reason that we presented in Subsection~\ref{subsec:fraccut}. 
But in principle, this value could be larger than the maximum flow value, for some $\URN$s. 
An advantage of the locality of the algorithms in this paper is that they nicely extend to $\URN$s, as we will see in this section. 
As a corollary, we obtain the Maximum Flow Minimum Cut Theorem for $\URN$s. 


We are ready to formulate the natural generalizations of the main results to unimodular random graphs and networks. 

\begin{Theorem}\label{localgflowuni}
Assuming a universal bound $M$ on the capacity function, for all $\eps > 0$, there is an $r\in \N$ and a local flow algorithm $F\: B_r^{(2)} \rightarrow \mathbb{R}$ such that for any unimodular random flow network $N$ we have
\begin{equation} \label{Thm1eqbounded}
\big\|\underline{F}(N)\big\| \ge \big\|\underline{f}^*(N)\big\| - \eps.
\end{equation}
\end{Theorem}

We indicate the necessary modifications to apply to the proof of Theorem~\ref{localg}. 
First, we need to slightly modify the assertion of Lemma~\ref{approx1}. 

\begin{Lemma} \label{approx1uni}
If the set of augmenting paths of length at most $\ell$ of a unimodular flow $f$ has $\sigma$ measure 0, then
\begin{equation} \label{approx1eq}
\big\|\underline{f}\big\| \ge \big\|\underline{f}^*\big\| - \frac{\underline{\capp}(N)}{\ell}.
\end{equation}
\end{Lemma}

The proof is essentially the same, except the summation should be replaced by expected value with respect to the probability measure $\sigma$. 
Lemma~\ref{lem:nonewshort} can be adjusted similarly. 

\begin{Lemma}\label{lem:nonewshortuni}
If the set of augmenting paths $u$ with $\len(u)<\ell$ of a flow $f$ has $\sigma$ measure 0, then augmenting on a path of length $\ell$ only creates a 0-measure set of new augmenting paths of length at most $\ell$.
\end{Lemma}

Chains of short paths and their label are defined for $\URG$s in the same way as before. 
There is a slight technical difficulty with the natural generalization of Lemma~\ref{spread} to $\URN$s.

\begin{Lemma} \label{spreaduni}
For each $\ell \in \mathbb{N}$ and $\eps > 0$ there exists a $q = q(\ell, \eps) \in \mathbb{N}$ such that for every unimodular random graph $G$ and $\sigma$-almost all undirected edges $e$ of $G$, 
the probability over all labelings that there exists a chain $u_1, u_2, \ldots, u_q$ for which $u_q$ contains $e$ is at most $\eps$.
\end{Lemma}
\begin{proof}
Note that the existence of such a chain for a given edge is determined by the $q\ell$-neighborhood of the edge. 
The involution invariant measure induces a probability distribution of such neighborhoods (with a distinguished root edge). 
By the Law of Total Probability, the probability described in the assertion of the lemma is a convex combination of the analogous probabilities on these finitely many finite neighborhoods. 
According to the finite variant, Lemma~\ref{spread}, for each of the neighborhoods the probability in question is at most $\eps$ (by choosing the same $q$ as in the proof of Lemma~\ref{spread}), and consequently, so is their convex combination. 
\end{proof}

The next step in the original proof was to define the algorithms $f_1(N,x)$ and $f_2(N,x)$. 
Once again, slight adjustments are needed. 
The main difference to the finite case is that it is not necessarily possible to deal with all short paths, since there are infinitely many of them in a $\URN$. 
So the definition of $f_1(N,x)$ makes use of the labelling $x$. 
Informally, the idea would be to systematically go through the short paths in the order of their label. 
This would not be a problem in finite networks; however, in an infinite network, the labels usually do not form a well-order. 
Hence, we could not even make sure that every short path is considered in a finite number of steps. 
As we want to define a local algorithm anyway, which can be applied in neighborhoods in parallel, it is more logical to consider a short path $u$, and to try to determine how the flow is altered on it (starting from the empty flow, and applying the Edmonds-Karp algorithm). 
In order to compute this, we need to know how the flow values were altered on short paths intersecting $u$ that have smaller label, as they had priority over $u$. 
Similarly, we need to know how the flow values were altered on the short paths intersecting the ones we just considered, etc. 
Note that the branches of this downward tree of intersecting short paths are chains. 
Thus Lemma~\ref{spreaduni} tells us that the probability that the process survives at least $q=q(\ell, \eps)$ rounds is at most $\eps$ for almost all initial $u$. 
In particular, if we are willing to run the process without a bound on the number of steps, then it halts with probability 1 (where the probability space is the product space of the seeds). 
This yields a factor of i.i.d.\ process, which is $f_1(N,x)$: it augments on $\sigma$-almost all short paths with probability 1. 
Note that $f_1(N,x)$ is not a local algorithm, as we did not set a bound on the depth $q$ above: setting such a bound is equivalent to allowing the process to work in a bounded neighborhood of $u$, that is, locality. 
This is exactly why we define $f_2(N,x)$ analogously to the finite case. 
Namely, if the above process did not vanish after a certain fixed number of steps $s$, then we give up trying. 
So we do not augment on those paths that can be obtained as the last element of a chain of length $s$. 
Hence, $f_2(N,x)$ is a block factor of i.i.d.\ process. 

The proof that $f_2(N,x)$ is local can be copied verbatim from the finite case, see Lemma~\ref{f2local}. 
The only alteration to be made in the statement is that it holds for $\sigma$-almost all edges. 
Similarly, Lemma~\ref{diff} holds with the same modification for $\URN$s, along with the rest of the proof of Theorem~\ref{localg}, concluding the proof of Theorem~\ref{localgflowuni}. 
We emphasize again here that taking the expectation of local flow algorithms with respect to all labelings yields a local flow algorithm, as the convex combination of flows is a flow. 

Next we generalize Proposition~\ref{localgfraccut}. 

\begin{Proposition}\label{localgfraccutuni}
Assuming a universal bound $M$ on the capacity function, for each $\eps > 0$ there exists a local fractional cut algorithm such that for each unimodular random flow network $N$, it produces a factor fractional cut with value $\big\|\underline{\tilde{C}}(N)\big\| \le \|\underline{f}^*\| + \eps$.
\end{Proposition}

Note that in the proof of Proposition~\ref{localgfraccut}, we first applied the local flow algorithm with some error $\eps_1$ and defined the residual graph. 
These local constructions lift to $\URN$s directly, taking expectations with respect to $\sigma$ over the finitely many possible neighborhoods as before. 
Then we defined a fractional cut based on the distance of a vertex from the set of sources $S$. 
As there are infinitely many sources in a $\URN$ in general, we need to be careful here. 
In the definition, the distance is only relevant as long as it is at most $1/\eps_2$ and $v\in R$. 
Hence, we can do it from the perspective of the vertex $v\in R$, scanning the $1/\eps_2$ neighborhood of it, which is once again possible in a $\URN$. 
If we find a source in this neighborhood, then we choose the one closest to $v$, and compute $\tilde{X}(v)$. 
If there is no source in this neighborhood, then we set $\tilde{X}(v)=1$. 
The rest of the proof translates directly to $\URN$s, by switching summations to expectations with respect to $\sigma$. 

We obtain the analog of Corollary~\ref{cor:integerflow} and Theorem~\ref{localgcut} just as in the finite setup. 

\begin{Corollary}\label{cor:integerflowURN}
Let $x$ be a mapping that assigns an independent uniform random seed in $[0,1]$ to the vertices of a unimodular random flow network. 
Assuming a universal bound $M$ on the capacity function and that all capacities are integers, for all $\eps > 0$, there is an $r\in \N$ and a random local flow algorithm $F\: B_r^{(2)} \rightarrow \mathbb{R}$ defined on the seeded neighborhoods such that 
\begin{equation} \label{Thm1eqbounded}
\E_x\big(\big\|\underline{F}(N)\big\|\big) \ge \|\underline{f}^*\| - \eps.
\end{equation}
and such that all values of the output flow are integers. 
\end{Corollary}

\begin{Theorem}\label{localgcutuni}
Assuming a universal bound $M$ on the capacity function, for all $\eps > 0$, there is an $r\in \N$ and a probability distribution $\D$ of local cut algorithms $C\: B_r \rightarrow \{0,1\}$ that each yield a factor cut such that for any unimodular random flow network $N$ we have
\begin{equation*}
\E_{C \in \D} \big\|\underline{C}(N)\big\| \le \big\|\underline{f}^*(N)\big\| + \eps.
\end{equation*}
\end{Theorem}

As a corollary, we obtain the highly nontrivial generalization of the Maximum Flow Minimum Cut Theorem for unimodular random flow networks. 

\begin{Theorem}[MFMC for URNs]
We assume a universal bound $M$ on the capacity function. 
Then the maximum flow ratio $\big\|\underline{f}^*(N)\big\|$, the minimum cut ratio, and the minimum fractional cut ratio coincide for any unimodular random flow network $N$. 
If all capacities are integers, then there is a maximum flow with all integer values. 
Except for the maximum integer flow problem, each of these optimization problems have an optimum that can be approximated up to an arbitrarily small error by the result of a factor process. 
The maximum integer flow is approximable by results of $\FIID$ processes. 
\end{Theorem}\label{thm:MFMCuni}
\begin{proof}
By compactness, maximum flows, minimum cuts and minimum fractional cuts exist. 
According to Proposition~\ref{localgfraccutuni} and Theorem~\ref{localgcutuni}, the three corresponding extremal values are equal. 
Similarly, by compactness,  maximum integer flows exist. 
According to Corollary~\ref{cor:integerflowURN}, provided that all capacity values are integers, the maximum integer flow ratio must coincide with the other three extremal ratios. 
\end{proof}

As stated in the introduction, Lovász has proved a much more general Maximum Flow Minimum Cut Theorem in~\cite{Lovaszmeasure}, for flows and cuts defined on Markov spaces, verifying the existence of a measurable minimum cut. 
The proof uses Zorn's lemma, and hence it is non-constructive. 
In this context, a unimodular random flow network (without the edge capacities) is a reversible Markov space whose underlying standard Borel space $J$ consists of the $\{S,R,T\}$-decorated countable rooted graphs with maximum degree $d$. 
The capacity measure on $J\times J$ disintegrates at each $x\in J$ into a linear combination of Dirac measures with at most $d$ summands, corresponding to the neighbors of the rooted graph obtained by moving the root to the other endpoint of an incident edge. 
Hence, the results in \cite{Lovaszmeasure} show that in $\URN$s, a measurable minimum cut, and in fact a measurable maximum flow exist with coinciding values. 
The present paper was also motivated by Lov\'asz' interest in the subject, who was Endre Cs\'oka's PhD supervisor at the time. 
The added value in the special case of $\URN$s is a constructive proof, and a local algorithm presented that finds a near-optimal structure. 
It is unknown whether such a local algorithm, in any meaningful sense of the expression, can be provided in the general setup, or in fact anywhere beyond $\URN$s with bounded degrees. 

\subsection{Connection with the Aldous-Lyons Conjecture}

Finally, we generalize Theorem~\ref{separation2} for $\URG$s, and show how the most fundamental conjecture about $\URG$s, due to Aldous and Lyons, can be strengthened. 
Roughly speaking, the conjecture says that the set of $\URG$s is the completion of the set of finite graphs. 
For more details cf. \cite{AlLy, Lovasz}. 

Let $\U$ denote the set of unimodular random graphs $\sigma \in \U$ with maximum degree at most $d$. 
Let $s_r(\sigma)$ denote the distribution of the $r$-neighborhoods of the root. 
Extending the definition in \eqref{Ddef}, let $S_r(\U) = cl\big\{s_r(\sigma) \mid \sigma \in \U \big\}= \big\{s_r(\sigma) \mid \sigma \in \U \big\}$. 
Since every graph $G$ with a uniform random root provides a $\sigma \in \U$ with $s_r(G) = s_r(\sigma)$, we have the trivial inclusion $S_r(\G) \subseteq S_r(\U)$. 
A sequence of $\URG$s $(\sigma_n)_{n\in \N}$ locally converges to the $\URG$ $\sigma$ if for every $r\in \N$ the neighborhood statistics converge, that is, $s_r(\sigma_n)\rightarrow s_r(\sigma)$. 
This coincides with the notion of convergence in the topology on $\U$ defined earlier. 
Using the identification of each finite graph (assuming a given degree bound) with the $\URG$ where the root is chosen uniformly at random, we can talk about sequences of finite graphs locally converging to a $\URG$. 
Note that the topology induced by this convergence notion is metrizable, for instance by the \emph{sampling distance};  see \cite{Lovaszbook} for more details. 

\begin{Conjecture}[Aldous--Lyons]
Every $\URG$ is the local limit of a sequence of finite graphs, or equivalently,
\begin{equation*}
\forall r \in \N\:\ S_r(\G) = S_r(\U).
\end{equation*}
\end{Conjecture}

Note that any $\MLA$ or $\RLA$ can be applied to $\URG$s. 
Thus a major problem in the previous subsection where we generalized the results from finite graphs to $\URG$s would have indicated that the Aldous-Lyons Conjecture might be false. 

We mention that there is a conflicting conjecture in group theory. 
A finitely generated group (with any given finite generating set) is called sofic if its Cayley graph viewed as a unimodular random graph is the local limit of a sequence of finite graphs. 
It is widely believed that there exist finitely generated non-sofic groups\footnote{To be precise, the Sofic Group Conjecture says that there is a non-sofic group. However, it is well-known that a group $G$ is sofic if and only if all finitely generated subgroups of $G$ are sofic. Thus it is an equivalent form of the conjecture that there exists a finitely generated non-sofic group.}, which would imply that the Aldous-Lyons Conjecture is false. 

\begin{Def}\label{def:AL}
We say that the Aldous--Lyons Conjecture is \textbf{completely false} if for all $\eps > 0$, there exists an $r \in \N$, a subset $B' \subset B_r$ and a $\URG$ $\sigma \in \U$ such that for all $G \in \G$, the $r$-neighborhood of a random vertex of $G$ is in $B'$ with probability at most $\eps$, but the $r$-neighborhood of the root of $\sigma$ is in $B'$ with probability at least $1 - \eps$.
\end{Def}

\begin{Theorem}\label{thm:strAL}
If the Aldous--Lyons Conjecture is false, then it is completely false.
\end{Theorem}

The proof once again follows the finite case. We show the natural analog of Theorem~\ref{separation2}, which directly implies Theorem~\ref{thm:strAL}. 

\begin{Theorem} \label{separation2uni}
Assume that $s_r(\sigma_0) \notin S_r(\F)$ holds for a nice family $\F$ of graphs, a unimodular random graph $\sigma_0 \in \U$, and $r\in \N$. 
Namely, there exists a $w_0\: B_r \rightarrow [0, 1]$ satisfying
\begin{equation} \label{deltadef}
\tilde{w}_0\big(s_r(\sigma_0)\big) - m(\F, w_0) = \delta > 0.
\end{equation}
Then for all $\eps > 0$, there is an $r' = r'(r, \eps, \delta)$, a subset $B' \subset B_{r'}$ and a factor spanned subgraph $\sigma_1$ of $\sigma_0$ such that 
\begin{align*}
\P\big(s_{r'}(\sigma_1) \in B'\big) &> 1 -\eps, \,\,and\\
\forall G \in \F\: \,\, \P\big(s_{r'}(G) \in B'\big) &< \eps. 
\end{align*}
\end{Theorem}

For the proof, we need that the class of $\URG$s that can be written as the local limit of a sequence of finite graphs is closed under taking factor subgraphs. 
We prove a somewhat more general statement. 
Regardless of this Aldous-Lyons conjecture, it asserts that the set of $\URG$s that occur as a local limit of a finite graph sequence is closed under taking $\FIID$ spanned subgraphs. 

\begin{Lemma}
Let $\sigma$ be a $\URG$ that is the local limit of a sequence of finite graphs, and let $\tau\subseteq \sigma$ be an $\FIID$ spanned subgraph of $\sigma$. 
Then $\tau$ is the local limit of a sequence of finite graphs. 
\end{Lemma}
\begin{proof}
Given an $\eps>0$, the $\FIID$ spanned subgraph $\tau\subseteq \sigma$ can be approximated up to an error $\eps$ by a constant-radius factor of i.i.d. spanned subgraph $\tau(\eps)\subseteq \sigma$. 
That is, $\tau(\eps)$ is an $\FIID$ spanned subgraph of $\sigma$ induced by a measurable function from the Bernoulli graphing to $\{0,1\}$ such that the value of the measurable function only depends on the seeded $r(\eps)$-neighborhood of the root for some $r(\eps)\in \N$, and the $L_1$-distance of this measurable function and the one that induces $\tau$ is less than $\eps$. 
This makes the total variation distance of $\tau(\eps)$ and $\tau$ less than $\eps$, and in particular, the sampling distance less than $\eps$; see formulas (19.1) and (19.4) in \cite{Lovaszbook}. 

The rule that maps 0 or 1 to a vertex depending on its $r$-neighborhood can be applied to vertices of finite graphs. 
Let $(G_n(\eps))_{n\in \N}$ be the sequence of finite graphs locally converging to $\sigma$ with $|V(G_n(\eps))|\geq n$. 
For each $n$, we randomly assign independent uniform seeds from $[0,1]$ (and also independently for all $n$) to the vertices of $G_n(\eps)$. 
Then the above block factor of i.i.d.\ rule computes a subgraph $H_n(\eps)\subseteq G_n(\eps)$. 
By local convergence, as $n\rightarrow \infty$ the $r(\eps)$-neighborhood statistics of $H_n(\eps)$ tends to that of $\tau(\eps)$ with probability 1, making $(H_n(\eps))_{n\in \N}$ locally converging to $\tau(\eps)$ with probability 1. 
(Clearly $|V(H_n(\eps))|\rightarrow \infty$ with probability 1.)

For each $k\in\N$ let $\eps_k=1/k$, and let $n(k)\in \N$ be an index so that $|V(H(\eps_k))_{n(k)}|>|V(H(\eps_{k-1}))_{n(k-1)}|$ and the sampling distance of $H(\eps_k)_{n(k)}$ and $\tau(\eps_k)$ is less than $\eps_k$. 
In particular, the sampling distance of $H(\eps_k)_{n(k)}$ and $\tau$ is less than $3\eps_k$. 
Hence, the sequence $(H(\eps_k)_{n(k)})_{k\in \N}$ locally converges to $\tau$.  
\end{proof}

The proof of Theorem~\ref{separation2uni} follows that of Theorem~\ref{separation2}, with the main ingredient being the averaging network $A(\sigma, r, w, \alpha)$ defined analogously to $\URN$s. 
Note that $\tilde{w}(s_r(\sigma)) = \mathbb{E}_\sigma(w(B_r(o)))$, where $o$ is the root. 
The slight technical difficulty arising from the graphs not being finite is that it is unclear whether supremal (or infimal) $\URG$s exist among the factor spanned subgraphs of $\sigma_0$. 
There is an easy fix which would have been unnecessary to include in the proof of Theorem~\ref{separation2}: we define $\eps$-supremal (and $\eps$-infimal) $\URG$s, meaning that it has no factor spanned subgraph such that the $\tilde{w}(s_r(.))$ value of the subgraph is more than $\eps$ bigger (smaller) than the $\tilde{w}(s_r(.))$ value of the given $\URG$. 
Whenever we talk about cuts in the proof, it should be replaced by factor cuts. 
Similarly, in the finite case we chose $G_1$ as a subgraph of $G_0$ with smallest $\tilde{w_1}(s_{r_1}(.))$ value, and $G_2$ as a subgraph with largest $\tilde{w_0}(s_r(.))$ value, but now we only consider factor subgraphs, instead. 
Once again, these extremal objects may not exist, so we choose factor spanned subgraphs that are close to the infimum/supremum, instead. 
These extremality conditions were applied by using the fact that in the averaging network, a carefully chosen cut can be close to optimal. 
As the cut algorithm is a mixed factor process, and a nearly optimal one can be constructed by a factor process (Theorem~\ref{thm:MFMCuni}), we lose nothing by the additional condition that only factor cuts are considered. 
The rest is a direct translation, changing summations to $\mathbb{E}_\sigma$ and probabilities to $\mathbb{P}_\sigma$. 

We can switch the quantifiers in Theorem~\ref{thm:strAL} in an advantageous way: if for every $\varepsilon>0$ there is an extreme counterexample, then there is an extreme counterexample that works for every $\varepsilon>0$. 

\begin{Corollary}
If the Aldous--Lyons Conjecture is false, then there is a $\URG$ $\sigma$ such that for all $\eps > 0$, there exists an $r \in \N$ and a subset $B' \subset B_r$ such that for all $G \in \G$, the $r$-neighborhood of a random vertex of $G$ is in $B'$ with probability at most $\eps$, but the $r$-neighborhood of the root in $\sigma$ is in $B'$ with probability at least $1 - \eps$.
\end{Corollary}
\begin{proof}
By Theorem~\ref{thm:strAL}, we may assume that the conjecture is completely false. 
Let $\eps(k)=1/k$. 
Then there is an $r(k)\in \N$, a $B'(k)\subset B_{r(k)}$, and a $\URG$ $\sigma_k$ as in Definition~\ref{def:AL}. 
We may assume that the sequence of numbers $r(k)$ is strictly monotone increasing. 
As the space of $\URG$s is compact with respect to the total variation distance, there is a convergent subsequence of the $\sigma_k$ tending to some $\URG$ $\sigma$. 

Given an $\eps>0$, let $n$ be large enough so that for any $k\geq n$ the TV distance of $\sigma_k$ and $\sigma$ is less than $\eps/2$ and also $1/n<\eps/2$. 
We claim that $r=r(n)$ and $B'=B'(n)$ is an appropriate choice. 
For any $G \in \G$, the $r(n)$-neighborhood of a random vertex of $G$ is in $B'(n)$ with probability at most $1/n$, which is less than $\eps/2$. 
On the other hand, the $r$-neighborhood of the root of $\sigma_n$ is in $B'(n)$ with probability at least $1 - 1/k> 1-\eps/2$. 
As the TV-distance of $\sigma$ and $\sigma_n$ is less than $\eps/2$, we have that the $r$-neighborhood of the root in $\sigma$ is in $B'(n)$ with probability at least $1 - \eps$. 
\end{proof}

The same method can be applied to strengthen the Sofic Group Conjecture. 
Namely, if a finitely generated group is non-sofic, then it must be very far from being the local limit of finite graphs. 

\begin{Corollary}
Given any finitely generated non-sofic group, then its Cayley graph is a $\URG$ $\sigma$ such that for all $\eps > 0$, there exists an $r \in \N$ such that for all $G \in \G$, the $r$-neighborhood of a random vertex of $G$ is isomorphic to the (unique) $r$-neighborhood of the root in $\sigma$ with probability at most $\eps$.
\end{Corollary}
\begin{proof}
We can copy all the arguments above. 
The only thing to note is that at a point where we would switch the $\URG$ to a (nearly) supremal subgraph defined by a factor process, we cannot actually change $\sigma$. 
Indeed, Cayley graphs are transitive, so any factor subgraph of such a graph contains either all vertices or none of them. 
This yields the following direct corollary: 
there exists an $r \in \N$ and subset $B' \subset B_r$ such that for all $G \in \G$, the $r$-neighborhood of a random vertex of $G$ is in $B'$ with probability at most $\eps$, but the $r$-neighborhood of the root of $\sigma$ is in $B'$ with probability at least $1 - \eps$. 
However, as Cayley graphs are transitive, the $r$-neighborhood of the root is concentrated on a single $r$-neighborhood. 
\end{proof}

\section{Acknowledgement}

The authors are indebted to L\'aszl\'o Lov\'asz for his help in making this paper, including the proposition of the main problems themselves.

\bibliographystyle{plain}
\bibliography{refs}

\begin{thebibliography}{10}

\bibitem{AlLy}
David Aldous and Russell Lyons.
\newblock {Processes on Unimodular Random Networks}.
\newblock {\em Electronic Journal of Probability}, 12:1454 -- 1508, 2007.

\bibitem{Angluin}
Dana Angluin.
\newblock Local and global properties in networks of processors (extended
  abstract).
\newblock In {\em Proceedings of the Twelfth Annual ACM Symposium on Theory of
  Computing}, STOC '80, pages 82--93, New York, NY, USA, 1980. Association for
  Computing Machinery.

\bibitem{BeSchSh}
Itai Benjamini, Oded Schramm, and Asaf Shapira.
\newblock Every minor-closed property of sparse graphs is testable.
\newblock In {\em Proceedings of the fortieth annual {ACM} symposium on
  {T}heory of computing}, STOC '08, pages 393--402. {ACM}, 2008.

\bibitem{BObT}
Andrej Bogdanov, Kenji Obata, and Luca Trevisan.
\newblock A lower bound for testing 3-colorability in bounded-degree graphs.
\newblock In {\em Proceedings of the 43rd Symposium on Foundations of Computer
  Science}, FOCS '02, pages 93--102, USA, 2002. IEEE Computer Society.

\bibitem{Csoka}
Endre Cs\'oka.
\newblock Comparison of variants of local algorithms.
\newblock https://arxiv.org/abs/1202.1565, 2012.

\bibitem{CzuShaSo}
Artur Czumaj, Asaf Shapira, and Christian Sohler.
\newblock Testing hereditary properties of nonexpanding bounded-degree graphs.
\newblock {\em SIAM Journal on Computing}, 38(6):2499--2510, 2009.

\bibitem{Elek2}
G{\'a}bor Elek.
\newblock Note on limits of finite graphs.
\newblock {\em Combinatorica}, 27:503--507, 2007.

\bibitem{Elek}
G{\'a}bor Elek.
\newblock Parameter testing in bounded degree graphs of subexponential growth.
\newblock {\em Random Structures \& Algorithms}, 37:248--270, 2010.

\bibitem{FoFu}
Lester~Randolph Ford and Delbert~Ray Fulkerson.
\newblock Maximal flow through a network.
\newblock {\em Canadian Journal of Mathematics}, 8:399--404, 1956.

\bibitem{LOW}
Christoph Lenzen, Yvonne~Anne Oswald, and Roger Wattenhofer.
\newblock What can be approximated locally?: case study: dominating sets in
  planar graphs.
\newblock In {\em The twentieth ACM Symposium on Parallel Algorithms and
  Architectures}, pages 46--54, 2008.

\bibitem{Linial}
Nathan Linial.
\newblock Locality in distributed graph algorithms.
\newblock {\em SIAM Journal on Computing}, 21(1):193--201, 1992.

\bibitem{Lovasz}
L\'aszl\'o Lov\'asz.
\newblock Very large graphs.
\newblock {\em Current Developments in Mathematics}, 2008:67--128, 2008.

\bibitem{Lovaszbook}
L\'aszl\'o Lov\'asz.
\newblock {\em Large Networks and Graph Limits}, volume~60 of {\em Colloquium
  Publications}.
\newblock American Mathematical Society, 2012.

\bibitem{Lovaszmeasure}
L\'aszl\'o Lov\'asz.
\newblock Flows on measurable spaces.
\newblock {\em Geometric and Functional Analysis}, 31:402--437, 2021.

\bibitem{MaRo}
Sharon Marko and Dana Ron.
\newblock Approximating the distance to properties in bounded-degree and
  general sparse graphs.
\newblock {\em ACM Transactions on Algorithms}, 5(2):1--28, 2009.

\bibitem{NaSt}
Moni Naor and Larry Stockmeyer.
\newblock What can be computed locally?
\newblock {\em SIAM Journal on Computing}, 24(6):1259--1277, 1995.

\bibitem{NgOn}
Huy~N. Nguyen and Krzysztof Onak.
\newblock Constant-time approximation algorithms via local improvements.
\newblock In {\em Proceedings of the 2008 49th Annual IEEE Symposium on
  Foundations of Computer Science}, FOCS '08, pages 327--336, USA, 2008. IEEE
  Computer Society.

\bibitem{Suomela}
Jukka Suomela.
\newblock Survey of local algorithms.
\newblock {\em ACM Computing Surveys}, 45(2):1--40, 2013.

\end{thebibliography}

\end{document}